\documentclass[hidelinks, a4paper]{article}
\usepackage[utf8]{inputenc}
\usepackage[english]{babel}
\usepackage[T1]{fontenc}

\usepackage{mathtools}
\usepackage{amsmath, amsthm, amssymb,amsfonts}
\usepackage{thmtools,thm-restate}

\usepackage{hyperref}
\usepackage[nameinlink,noabbrev]{cleveref}

\usepackage{graphicx}
\usepackage{float}
\usepackage{caption}
\usepackage[shortlabels]{enumitem}

\usepackage{natbib} 
\usepackage[a4paper, total={16cm, 22cm}]{geometry}

\newtheorem{theorem}{Theorem}
\newtheorem{proposition}[theorem]{Proposition}

\newtheorem{definition}[theorem]{Definition}
\newtheorem{lemma}[theorem]{Lemma}
\newtheorem{corollary}[theorem]{Corollary}

\newtheorem{example}{Example}

\newcommand{\BE}{{\mathbb{E}}}

\newcommand{\BP}{{\mathbb{P}}}
\newcommand{\BQ}{{\mathbb{Q}}}
\newcommand{\BR}{{\mathbb{R}}}

\newcommand{\CE}{{\cal E}}

\newcommand{\CN}{{\cal N}}

\newcommand{\CQ}{{\cal Q}}

\newcommand{\mX}{{{m_x}}}
\newcommand{\mH}{{{m_h}}}
\newcommand{\mA}{{{m_A}}}

\newcommand{\dd}{{\mathrm{d}}}

\DeclareMathOperator{\Var}{Var} 
\DeclareMathOperator{\Cov}{Cov}

\DeclareMathOperator*{\argmin}{argmin}

\DeclareMathOperator*{\plim}{plim}

\DeclareMathOperator{\cvar}{CVaR}

\newcommand{\Id}{\mathrm{Id}}

\newcommand{\toP}{\overset{\BP}{\to}}
\newcommand{\indep}{\perp \!\!\! \perp}

\let\phi\varphi
\let\theta\vartheta
\let\epsilon\varepsilon

\let\le\leqslant
\let\ge\geqslant
\let\leq\leqslant
\let\geq\geqslant

\def\keywordname{{\bfseries \emph{Keywords}}}%
\def\keywords#1{\par\addvspace\medskipamount{\rightskip=0pt plus1cm
\def\and{\ifhmode\unskip\nobreak\fi\ $\cdot$
}\noindent\keywordname\enspace\ignorespaces#1\par}}

\title{CVaR anchor regression protects against rare shifts}
\date{October 2026}
\author{Malte Londschien}
\begin{document}

\maketitle

\begin{abstract}
\noindent We study prediction in new environments when the training data contains rare, large distribution shifts.
Anchor regression penalizes the average of the squared mean residual across environments.
It protects against shifts in an ellipsoid determined by the second moment of the training shifts.
Covering rare shifts may therefore require a large penalty $\gamma \gg 1$, expanding the ellipsoid in every direction and reducing accuracy on common environments.
We propose CVaR anchor regression, which replaces the average of the squared mean residuals with a tail average.
We prove an exact worst-case risk guarantee under a linear structural model that allows for heteroscedastic noise.
CVaR anchor regression protects against shifts in a union of ellipsoids determined by second moments of the shifts in the worst $\alpha$-fraction of environments.
For $\alpha$ small enough, this is the scaled convex hull of the training shifts and their negatives.
It therefore covers all training shifts when $\gamma \geq 1$.
Unlike CVaR or GroupDRO applied directly to prediction risks, CVaR anchor regression does not give environments more weight solely because their noise levels are high.
Examples show how the method can improve protection against rare shifts while retaining accuracy on common environments.
We illustrate the method on the New York City taxi dataset.
\end{abstract}
\section{Introduction}

Predictive models are often deployed on data whose distribution differs from the training distribution.
We are particularly interested in the setting where the training data come from many discrete environments (for example, patients from different hospitals, observations from different days or regions) and predictions must extrapolate to new, possibly more extreme environments.
In such settings, robustness is often assessed using upper quantiles or tail averages of the environment-wise prediction risk \citep{rothenhausler2021anchor,koh2021wilds,eastwood2022probable}.

Distributionally robust optimization (DRO) minimizes the worst-case risk over a set of distributions \citep{bental2013robust,duchi2021learning} obtained by reweighting.
Applied to environment-wise risks, one choice is to minimize their conditional value at risk \citep[CVaR,][]{rockafellar2000optimization}, also studied as superquantile risk minimization \citep[Appendix~F]{eastwood2022probable}.
This is the average risk in the worst $\alpha$-fraction of environments.
As $\alpha \to 0$, this becomes GroupDRO, which minimizes the largest environment-wise risk \citep{sagawa2020distributionally}.
Other proposals target a quantile of the environment-wise risks \citep{eastwood2022probable} or penalize their variance \citep[V-Rex,][]{krueger2021out}.

As these criteria act on total environment-wise risks, they can emphasize environments with high noise levels \citep{oren2019distributionally,agarwal2022minimax,wang2023distributionally}.
To see why, decompose the environment-wise risk as
$$
R_e(b) = \BE[(Y - X^Tb)^2 \mid \text{environment } e] = \BE[Y - X^Tb \mid e]^2 + \Var[Y - X^Tb \mid e] = \mu_e(b)^2 + \sigma_e^2(b).
$$
The squared mean residual $\mu_e(b)^2$ reflects how strongly environment $e$ is shifted. The variance $\sigma_e^2(b)$ reflects the noise level in environment $e$.
If the noise levels $\sigma_e^2(b)$ vary between environments, reweighting environments according to their total risk emphasizes noisy environments, inflating variance without improving robustness.

Anchor regression \citep{rothenhausler2021anchor} treats these two components separately.
For discrete environments $e \in \CE$ with relative weights $w_e$, it is given by
$$
b_\mathrm{anchor}^\gamma := \argmin_b \sum_{e \in \CE} w_e \sigma_e^2(b) + \gamma \sum_{e \in \CE} w_e \mu_e(b)^2.
$$
Increasing $\gamma$ strengthens the penalty on squared mean residuals while leaving the residual variance term unchanged.
Under their structural model, this protects against an ellipsoid of shifts whose shape is determined by the second moment of the shifts in the training data.
This guarantee also holds with heteroscedastic noise (\cref{thm:anchor-regression}). The ellipsoid scales by $\sqrt{\gamma}$ without changing its shape (first panel of \cref{fig:example_1}).

In practice, the distribution of shifts might not be elliptical and their shape not be well captured by their second moment.
Covering rare training shifts then requires a large $\gamma$, expanding the ellipsoid in every direction and reducing performance on common environments.
We combine DRO reweighting with anchor regression by replacing the average of the squared mean residuals with their CVaR.
This gives CVaR anchor regression:
\begin{align*}
b_\mathrm{cvar-anchor}^{\gamma, \alpha} &:= \argmin_b \sum_{e \in \CE} w_e \sigma_e^2(b) + \gamma \cvar_\alpha(\mu_e(b)^2) \\
&= \argmin_b \sum_{e \in \CE} w_e \sigma_e^2(b) + \gamma \max \big\{ \sum_{e \in \CE} q_e \mu_e(b)^2 \mid \sum_{e \in \CE} q_e = 1, 0 \leq q_e \leq w_e / \alpha  \big\} .
\end{align*}
For $\alpha = 1$, this recovers ordinary anchor regression.
As in anchor regression, $\gamma$ controls the scale of the robustness set.
Decreasing $\alpha$ enlarges the set beyond the anchor regression ellipsoid (\cref{thm:cvar-anchor-robustness}).
For finitely many environments and sufficiently small $\alpha$, it equals the scaled convex hull of the training shifts and their negatives, and contains every training shift when $\gamma \geq 1$ (\cref{prop:finite-convex-hull}).

\textbf{Related work.} \quad
Other approaches counteract the dependence on noise levels by comparing the risk in each environment with a reference risk that does not depend on the predictor.
Maximin effects maximize the smallest explained variance $\BE[Y^2 | e] - R_e(b)$ \citep{meinshausen2015maximin}, later extended to general machine learning models by \citet{wang2023distributionally}.
Minimax regret compares with the best predictor for each distribution $R_e(b) - \min_{b'} R_e(b')$  \citep{agarwal2022minimax,mo2024minimax}.
These reference risks cancel predictor-independent contributions to risk (for example, if the additive noise in the response variable is heteroscedastic).
However, if the covariates are heteroscedastic, the residual variance $\sigma_e^2(b)$ can depend on the predictor $b$, so a fixed reference does not remove it.
CVaR anchor regression instead reweights only the squared mean residuals $\mu_e(b)^2$ and retains the average residual variance $\sum_e w_e \sigma_e^2(b)$ in its objective.

\begin{example}
    \label{ex:1}
Let $X \in \BR^2$, $Y \in \BR$, with hidden confounder $H \in \BR^2$:
$$
X = \delta_e + H + \varepsilon_X, \qquad Y = H_1 + H_2 + \varepsilon_Y, \qquad H, \varepsilon_X \sim \CN(0, \Id_2), \ \ \varepsilon_Y \sim \CN(0, \sigma_e^2).
$$
Let $\delta_e = \pm (1, 0)^T$ and $\sigma_e^2 = 2$ for $e \in P_1$ and $\delta_e = \pm(0, 3)^T$ and $\sigma_e^2 = 1$ for $e \in P_2$.
The training distribution is $P_\mathrm{train} = 0.9 P_1 + 0.1 P_2$.
The left panel of \cref{fig:example_1} shows the robustness sets of CVaR anchor regression for varying $\gamma, \alpha$.
The second moment of shifts $\delta_e$ is $0.9 \, \Id_2$, so covering the rare shifts in $P_2$ with anchor regression requires $\gamma \geq 10$.
In contrast, CVaR anchor regression with $\alpha = 0.1$ covers the rare shifts with $\gamma = 1$.
The right panel shows the estimators' risks on $P_1$ and $P_2$.
\end{example}

\begin{figure}[H]
    \centering
    \includegraphics[width=0.9\textwidth]{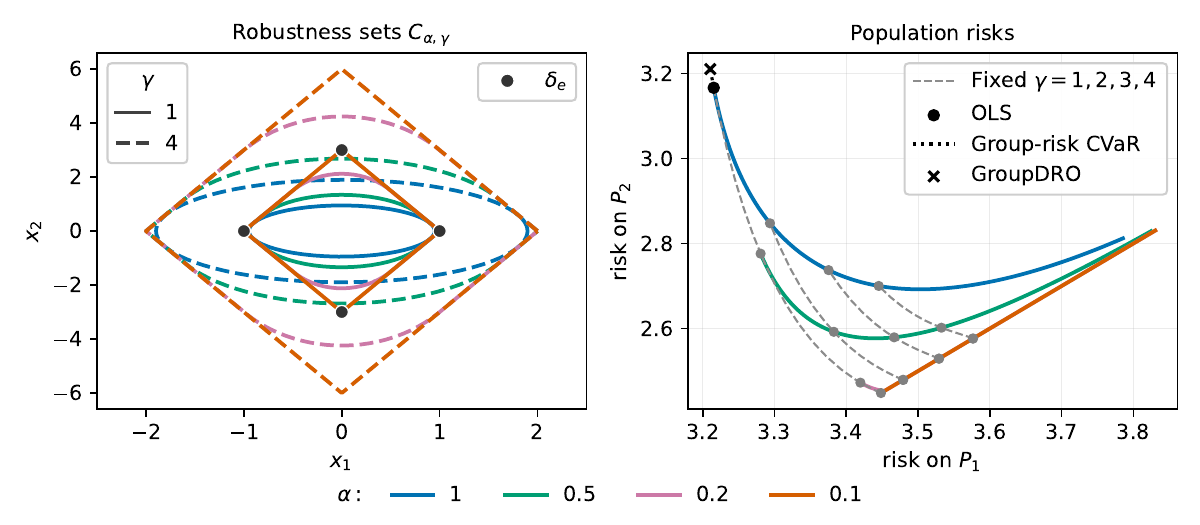}
    \caption{\label{fig:example_1}Left: Robustness sets of CVaR anchor regression $C_{\alpha, \gamma}$ for \cref{ex:1} varying $\gamma=1,4$ and $\alpha=1, 0.5, 0.2, 0.1$. The set $C_{0.1, 1}$ includes the training shifts $\pm(0,3)^T$, while $C_{1, 4}$ does not. Right: Population risks of CVaR anchor regression on the two training environments $P_1$ and $P_2$. Both increasing $\gamma$ and decreasing $\alpha$ result in a trade-off between the two environment risks, but simultaneously decreasing $\gamma$ and $\alpha$ can strictly improve performance on both environments.}
\end{figure}

\section{Population formulation and shift robustness}
\label{sec:structural-model}
We follow the notation of \citet{rothenhausler2021anchor}.
Let $B \in \BR^{(\mX + \mH + 1) \times (\mX + \mH + 1)}$ and $M \in \BR^{(\mX + \mH + 1) \times (\mA)}$ be fixed with $\mathrm{Id} - B$ invertible.
Let $\varepsilon \in \BR^{\mX + \mH + 1}$ and $A \in \BR^\mA$ be random vectors with finite second moments, $\BE[\varepsilon] = 0$, and $\Cov[\varepsilon, A] = 0$.
Let the distribution of $(X, Y, H, A)$ under $\BP_\mathrm{train}$ be a solution of
\begin{equation}\label{eq:structural}
\begin{pmatrix}
X \\
Y \\
H
\end{pmatrix}
=
B \cdot \begin{pmatrix}
X \\
Y \\
H
\end{pmatrix} + \varepsilon + M A \quad \Leftrightarrow \quad \begin{pmatrix}
X \\
Y \\
H
\end{pmatrix} = (\mathrm{Id} - B)^{-1} (\varepsilon + M A).
\end{equation}
We are interested in the risk under interventions on $A$.
For possibly random $v$ independent of $\varepsilon$, let $\BP_v$ be the distribution of $(X, Y, H)$ a solution of \cref{eq:structural} with $M A$ replaced by $v$.
The law of $\varepsilon$ under $\BP_v$ is the same as under $\BP_\mathrm{train}$, that is, $\varepsilon$ is marginalized over $A$.
Write $P_A Z := \BE[Z A^T]\BE[A A^T]^{\dagger}A$ for the linear projection of $Z$ onto $A$ and $M_A Z = Z - P_A Z$ ( $^\dagger$ is the Moore--Penrose pseudoinverse).
Unless another law is specified, expectations and CVaR are taken under the training distribution.
All proofs can be found in appendix \ref{sec:proofs}.

\subsection{Anchor regression for heteroscedastic noise}
\citet{rothenhausler2021anchor} assume homoscedastic noise as part of their structural model.
This is not necessary for the robustness guarantee of anchor regression.
\begin{theorem}[Anchor regression, Theorem 1 of \citet{rothenhausler2021anchor} with uncorrelated, instead of independent $\varepsilon$ and $A$] \label{thm:anchor-regression}
Let $\gamma \geq 0$, $b \in \BR^\mX$, and define $C_\gamma := \{ v \mid v v^T \preceq \gamma M \BE[A A^T] M^T \}$.
Under the model of \cref{eq:structural},
$$
\BE[(M_A (Y - X^T b))^2] + \gamma \BE[ (P_A (Y - X^T b))^2] = \sup_{v \in C_\gamma} \BE_v[ (Y - X^T b)^2].
$$
\end{theorem}%

\Cref{thm:anchor-regression} motivates the anchor regression estimator
$$
b^\gamma_\mathrm{anchor} := \argmin_b \BE[(M_A (Y - X^T b))^2] + \gamma \BE[ (P_A (Y - X^T b))^2].
$$
Note how both the anchor penalty term $\BE[ (P_A (Y - X^T b))^2]$ and the robustness set $C_\gamma$ are constructed by taking averages over shifts with respect to the training distribution.

\subsection{CVaR anchor regression and shift robustness}
For any $\alpha \in (0, 1]$ and random variable $Z$ with c.d.f.\ $F_Z$, define $$\cvar_\alpha(Z) := 1 / \alpha \int_{1-\alpha}^1  F_Z^{-1}(u) du = \sup \{ \BE_\BQ[Z]  \mid \BQ \ll \BP \text{ is a prob. measure with } \| \dd \BQ / \dd \BP \|_\infty \leq 1/\alpha \}$$ to be the conditional value at risk.
If $Z$ is discrete, taking values $z_e$ with probabilities $p_e$, then
$$
\cvar_\alpha(Z) = \max \{ \sum_e q_e z_e \mid \sum_e q_e = 1, 0 \leq q_e \leq p_e / \alpha \}.
$$
We propose to replace the expectation in the anchor penalty with the conditional value at risk.

\begin{definition}[CVaR anchor regression]
    For $\gamma \geq 0$ and $\alpha \in (0, 1]$, the CVaR anchor regression estimator is
    $$
    b^{\gamma, \alpha}_\mathrm{cvar-anchor} := \argmin_b \BE[ (M_A(Y - X^T b))^2] + \gamma \cvar_\alpha ( (P_A (Y - X^T b))^2 ).
    $$
\end{definition}
\noindent This generalizes anchor regression, with $b_\mathrm{anchor}^\gamma = b^{\gamma, 1}_\mathrm{cvar-anchor}$.
\begin{theorem}[Robustness set of CVaR anchor regression] \label{thm:cvar-anchor-robustness}
      Fix $\alpha \in (0, 1]$ and $\gamma \ge 0$. Let
      $$
      \CQ_\alpha := \big\{ \BQ \ \mid \  \BQ \ll \BP_\mathrm{train} \text{ is a probabaility measure with } \| \dd \BQ / \dd \BP_\mathrm{train} \|_\infty \le 1/\alpha \big\}
      $$                              
      be the dual set of the $\mathrm{CVaR}_\alpha$ and let
      $$
      C_{\alpha, \gamma} := \bigcup_{\BQ \in \CQ_\alpha}
        \Big\{ v \ \Big| \ v v^T \preceq
          \gamma M \BE_\BQ[A A^T] M^T \Big\} .
      $$
      Then, for every $b \in \BR^{\mX}$,
      $$
      \sup_{v \in C_{\alpha, \gamma}} \BE_v\big[(Y - X^T b)^2\big]
      = \BE\big[(M_A(Y - X^T b))^2\big]
        + \gamma \mathrm{CVaR}_\alpha\big( (P_A (Y - X^T b))^2 \big)
      $$
      and hence $b_\mathrm{cvar-anchor}^{\gamma, \alpha}
        = \argmin_b \sup_{v \in C_{\alpha, \gamma}} \BE_v[(Y - X^T b)^2]$.
  \end{theorem}
An analogous result holds for random $v$ with $v v^T$ replaced by $\BE_v[v v^T]$ (see appendix \ref{app:proof-cvar-anchor-robustness}).

Recall that under $\BP_v$ the noise law is marginalized (averaged) over $A$.
Thus, covering a shift in the training data does not necessarily bound the risk of the environment with that shift, whose noise law may different from the training marginal.
Indeed, this is the main idea of CVaR anchor regression: It does not emphasize environments solely because their noise levels are high.
Instead, it protects against large shift contributions to the risk.
\subsection{Geometry of the robustness set}
\label{sec:geometry}
We present some geometric properties of the robustness sets $C_{\alpha, \gamma}$.
\begin{proposition}[Monotonicity of the robustness sets] \label{prop:nesting}
Let $\alpha, \alpha' \in (0,1]$ and $\gamma, \gamma' \ge 0$.
If $\gamma \leq \gamma'$ and $\gamma / \alpha \leq \gamma' / \alpha'$, then $C_{\alpha, \gamma} \subseteq C_{\alpha', \gamma'}$.
\end{proposition}
In particular, the set $C_{\alpha, \gamma}$ is monotone in each parameter if the other one is fixed and $C_{1,\gamma} \subseteq C_{\alpha,\gamma} \subseteq C_{1,\gamma/\alpha}$.

\begin{proposition}[Convex hull for discrete anchors and small $\alpha$] \label{prop:finite-convex-hull}
Suppose $A$ is discrete and $MA$ takes values $\delta_e$ with probabilities $w_e > 0$.
For $0 < \alpha \leq \min_e w_e$ and $\gamma \geq 0$,
$$
C_{\alpha, \gamma} = \sqrt{\gamma}\,\operatorname{conv}\{\pm\delta_e : e \in \CE\}.
$$
\end{proposition}

\subsection{Equivalences for elliptical shifts and homoscedastic noise}
\label{sec:equivalences}
\citet[Lemma 1]{rothenhausler2021anchor} show that for $(X, Y, A)$ centered Gaussian,
$$\BE[(M_A(Y-X^T b))^2] + \gamma \BE[(P_A(Y-X^T b))^2] = \mathrm{quantile}_{\alpha=F_{\chi^2_1}(\gamma)} \left( \BE[(Y-X^T b)^2 \mid A] \right).$$
That is, in this setting, anchor regression already minimizes a quantile of the conditional risks given the anchor.

We prove that for centered elliptical shifts, CVaR anchor regression is equivalent to anchor regression with a rescaled penalty.
Furthermore, for any shift distribution, CVaR anchor regression minimizes the CVaR of the conditional risk given the anchor in a modified model where the noise has its training marginal distribution and is independent of $A$, and the shifts are scaled by $\sqrt{\gamma}$.
Finally, in the homoscedastic case with $\BE[\varepsilon\mid A]=0$ and $\Cov(\varepsilon\mid A)$ constant, CVaR anchor regression minimizes the CVaR of risk given the anchor under the training distribution.
For discrete $A$ and with $\gamma=1$, this is the same as superquantile risk minimization \citep[Appendix~F]{eastwood2022probable} and for $\alpha \to 0$ it is the same as GroupDRO \citep{sagawa2020distributionally}.
\begin{lemma}
  \label{lem:elliptical}
  Under the model of \cref{eq:structural}, suppose $MA$ has a
  centered elliptical distribution. That is, there exists a
  univariate random variable $Z$ with $\BE[Z^2]=1$ such that,
  for every $u\in\BR^{\mX+\mH+1}$,
  $$u^TMA\stackrel{d}{=}  Z\sqrt{u^T\BE[MAA^TM^T]u}.$$
  Then, for any $b$,
  $$(P_A(Y-X^Tb))^2 \stackrel{d}{=} Z^2 \BE[(P_A(Y-X^Tb))^2].$$
  For any real-valued positive homogeneous functional $T\colon Z \mapsto T(Z)$ that only depends on the distribution of $Z$ (e.g., the CVaR or quantile), this implies
  $$T \left( (P_A(Y-X^Tb))^2 \right) = T(Z^2) \BE[(P_A(Y-X^Tb))^2].$$
\end{lemma}

\begin{corollary}[Equivalence for elliptical shifts] \label{prop:elliptical}
  \label{cor:elliptical}
Under the distributional assumptions of \cref{lem:elliptical}, let $\alpha \in (0, 1]$ and $c_\alpha:=\cvar_\alpha(Z^2) \in [1, 1/\alpha]$.
Then, for every $b$,
$$
\cvar_\alpha\big((P_A(Y-X^T b))^2\big)
= c_\alpha \, \BE\big[(P_A(Y-X^T b))^2\big]
$$
and thus $b_\mathrm{cvar-anchor}^{\gamma,\alpha} = b_\mathrm{anchor}^{\gamma c_\alpha}$.
\end{corollary}

\begin{proposition}[CVaR of conditional risks]
\label{prop:conditional-risks}
Fix $\gamma \geq 0$ and $\alpha \in (0,1]$.
Let $\tilde X, \tilde Y, \tilde H$ be a solution of \cref{eq:structural} but with $\varepsilon$ replaced by $\tilde\varepsilon$, which has the marginal training distribution of $\varepsilon$ and is independent of $A$ and $MA$ replaced by $\sqrt{\gamma} MA$.
Then, for every $b$,
$$
\cvar_\alpha(\BE[(\tilde Y - \tilde X^T b)^2 \mid A]) = \BE[(M_A(Y - X^T b))^2] + \gamma \cvar_\alpha((P_A(Y - X^T b))^2).
$$
\end{proposition}

\begin{proposition}[Equivalence to group CVaR]\label{prop:equivalence-group-cvar}
Suppose the structural model \eqref{eq:structural} holds with $\BE[\varepsilon | A]=0$ and $\Cov(\varepsilon | A)$ constant.
Then, for every $b$ and $\alpha \in (0,1]$,
$$
\cvar_\alpha\big(\BE[(Y-X^Tb)^2 \mid A]\big)
=
\BE[(M_A(Y-X^Tb))^2]
+
\cvar_\alpha\big((P_A(Y-X^Tb))^2\big).
$$
Consequently, population group CVaR and CVaR anchor regression with $\gamma=1$ have the same set of minimizers.
\end{proposition}

\section{Computation and practical considerations}
We use the plug-in estimator for the CVaR anchor regression coefficient:
$$
\hat b_\mathrm{cvar-anchor}^{\gamma, \alpha} := \argmin_b n^{-1} \| \hat M_A (Y - X b) \|_2^2 + \gamma \widehat{\mathrm{CVaR}}_\alpha( (\hat P_A (Y - X b))^2 ),
$$
where $\hat P_A = A (A^T A)^{\dagger} A^T$ and $\hat M_A = \mathrm{Id}_n - \hat P_A$ are the projections onto the column span $A$ and its orthogonal complement ($^\dagger$ is the Moore--Penrose pseudoinverse).
The square is componentwise and empirical CVaR assigns weight $1/n$ to each component.

For discrete environments, with slight abuse of notation, write $i\in e$ if the $i$-th sample belongs to environment $e$, let $n_e = |e|$ be the number of samples in environment $e$, and let $w_e = n_e / n$ be the empirical environment weights.
Then, for
$$
\hat\mu_e(b) := \frac{1}{n_e}\sum_{i \in e}(Y_i-X_i^T b), \quad
\hat\sigma_e^2(b) := \frac{1}{n_e}\sum_{i \in e}(Y_i-X_i^T b-\hat\mu_e(b))^2,
$$
we have
\begin{align*}
\hat b_\mathrm{cvar-anchor}^{\gamma, \alpha} &= \argmin_b \sum_{e \in \CE} w_e \hat\sigma_e^2(b) + \gamma \, \widehat{\mathrm{CVaR}}_\alpha(\hat\mu_e(b)^2) \\
&= \argmin_b \sum_{e \in \CE} w_e \hat\sigma_e^2(b) + \gamma \, \max\Big\{ \sum_{e \in \CE} q_e \hat\mu_e(b)^2 \mid q_e \ge 0, \, \sum_{e \in \CE} q_e = 1, \, q_e / w_e \le 1/\alpha \Big\}.
\end{align*}

\subsection{Asymptotic normality}
\begin{proposition}[Consistency and asymptotic normality]\label{prop:asymptotic-normality}
Let $\gamma \geq 0$ and $\alpha \in (0,1]$.
Under the model of \cref{eq:structural}, let
$(X_i,Y_i,A_i)_{i=1}^n$ be independent copies of $(X,Y,A)$, with fixed dimensions.
Assume $\BE[AA^T]$ is invertible and the population minimizer
$b^\star:=b_\mathrm{cvar-anchor}^{\gamma,\alpha}$ is unique.
Then
$$
\hat b_\mathrm{cvar-anchor}^{\gamma,\alpha} \overset{P}{\longrightarrow} b^\star.
$$
(consistency). Assume additionally finite fourth moments and one of the following:
\begin{enumerate}[(i)]
\item $\alpha=1$ or $\gamma=0$.
\item \emph{Smooth expected loss.} Let
$$
L(b,t)=\BE[(M_A(Y-X^Tb))^2]
+\gamma\left\{t+\frac1\alpha\BE\left[\left((P_A(Y-X^Tb))^2-t\right)_+\right]\right\},
$$
and let $t^\star$ minimize $L(b^\star,t)$.
The function
$$
(d,t)\longmapsto\BE\left[\left((A^Td)^2-t\right)_+\right]
$$
is twice continuously differentiable near
$\bigl(\BE[AA^T]^{-1}\BE[A(Y-X^Tb^\star)],t^\star\bigr)$, and
$\nabla_{(b,t)}^2L(b^\star,t^\star)\succ0$.
\item \emph{Finitely many environments with a stable cutoff.}
$A$ consists of indicators for finitely many environments with probabilities $p_e>0$.
After ordering the environments, their residual means $\mu_e(b)=\BE[Y-X^Tb\mid e]$ satisfy
$$
\mu_1(b^\star)^2>\cdots>\mu_m(b^\star)^2,
\qquad
\sum_{e<k}p_e<\alpha<\sum_{e\leq k}p_e
$$
for some $k$.
\end{enumerate}
Then, for some covariance matrix $V$,
$$
\sqrt n\left(\hat b_\mathrm{cvar-anchor}^{\gamma,\alpha}-b^\star\right)
\overset{d}{\longrightarrow}\CN(0,V)
$$
(asymptotic normality).
\end{proposition}
See appendix \ref{app:proof-asymptotic-normality} for the proof.
If $(X,Y,A)$ is jointly Gaussian, then $V$ admits a closed form expression as a function of population moments.
We compute this for a special case in \cref{ex:4}.
Assumptions (ii) and (iii) prohibit contrived cases where at $t^\star \in \mathrm{quantile}_{1-\alpha} ( (P_A (Y - X^T b^\star))^2 )$ there is an atom $\BP((P_A (Y - X^T b^\star)^2) = t^\star) > 0$.
Take for example $X=\varepsilon+\delta_e, \ Y=1-2\delta_e$ for $\varepsilon\sim\mathcal N(0,1)$ and two equally likely environments with $\delta_e = 0, 1$.
For $\gamma=1, \alpha=1/4$, $L(b) = b^2 + \max( 1, (1+b)^2 )$ with unique minimizer $b^\star=0$.
Let $\bar X_e$ be environment averages.
The empirical loss is $\hat L(b) \toP b^2 + \gamma \max( (1 - b \bar X_0)^2, (1 + b \bar X_1)^2 )$ with derivatives
$\frac{\dd}{\dd b} \hat L(b)|_{b=0+} \toP 2 \bar X_1 \toP 2$ and $\frac{\dd}{\dd b} \hat L(b)|_{b=0-} \toP - 2 \bar X_0$.
Thus $\hat b = \min \{ 0, \bar X_0 \} + o_P(n^{-1/2})$, where $\sqrt n X_0 \sim \CN(0,2)$.
That is, $\sqrt{n} \hat b$ converges to a censored Gaussian.

\subsection{Optimization}
By the Rockafellar--Uryasev representation $\cvar_\alpha(Z) = \min_{t \in \BR} \{ t + \BE[(Z - t)_+] / \alpha \}$ \citep{rockafellar2000optimization}, computing $\hat b^{\gamma,\alpha}_\mathrm{cvar\text{-}anchor}$ amounts to minimizing
$$
\| (M_A(Y - X b))^2 \|_2^2 / n+ \gamma \Big( t+\frac{1}{n\alpha}\sum_{i=1}^n \big([\hat P_A(Y-Xb)]_i^2-t\big)_+ \Big)
$$
jointly over $(b, t)$.
This is convex in $(b, t)$, but generally non-smooth in $t$ and the minimizing $t$ may be an interval.
To improve the convergence of gradient-based solvers, we suggest replacing the hinge $(\cdot)_+$ by the softplus $\tau \log(1 + \exp(\cdot / \tau))$ \citep[see also][Section~3.2.1]{soma2020statistical}.
We further suggest annealing $\tau \to 0$ over iterations.
This is implemented in \url{github.com/mlondschien/cvar-anchor-regression}.

\subsection{Finite sample bias}\label{sec:finite-sample-bias}
The plug-in estimate $\hat \mu_e(b)^2$ is biased for $\mu_e(b)^2$: $\BE[\hat \mu_e(b)^2] = \mu_e(b)^2 + \sigma_e^2(b) / n_e$.
If the $w_e = 1 / |\CE|$ are constant, this effectively bounds the effective penalty at $\gamma_\mathrm{effective} = n_e$ for $\gamma \to \infty$.
This applies to both standard anchor regression and CVaR anchor regression.
See also \cref{ex:4}.
However, for small or noisy environments this bias is larger and the CVaR tail may select them even when their true mean residuals are small.
One could counteract this by using a debiased estimate $\hat \mu_e(b)^2 - \hat \sigma_e^2(b) / (n_e - 1)$ instead, but the resulting objective is no longer convex in $b$.

\subsection{Inclusion of an intercept}
In standard anchor regression, the following are equivalent ways to include an intercept:
(i) centering $X$, $A$, and $y$, running the regression without an intercept, and then computing the intercept as $\bar y - \bar X^T \hat b$ and (ii) including a constant column in both $X$ and $A$.
For CVaR anchor regression, these are no longer equivalent.
We implement (ii) in \url{github.com/mlondschien/cvar-anchor-regression} as this corresponds to the CVaR penalty over the squared group-wise mean residuals $\hat \mu_e(b)^2$ instead of the square of those relative to the overall mean residual $(\hat \mu_e(b) - \sum_e w_e \hat \mu_e(b))^2$.
The choice between (i) and (ii) remains a modeling choice.
See appendix \ref{sec:intercepts} for a more detailed discussion on how to include an intercept.

\section{Examples}
\Cref{ex:1} in the introduction presents a setup where noises are heteroscedastic and shifts are non-elliptical.
CVaR anchor regression improves upon standard anchor regression and groupDRO.
We provide two more illustrative examples below, followed by two sampling variance calculations for Gaussian shifts.
\Cref{ex:2} presents an anticausal model with heteroscedastic noise in $X$ and non-elliptical shifts.
In \cref{ex:3}, the shifts are uniformly distributed, there is no hidden confounder, and the additive noise is homoscedastic, but the model is misspecified with $Y = 0.1 X^3 + \varepsilon_Y$.

\begin{example}[Anticausal model]
\label{ex:2}
Let
$$
Y = \varepsilon_Y \sim \CN(0, 1), \quad X = 2 Y \begin{pmatrix} 1 \\ 1 \end{pmatrix} + \delta_e + \varepsilon_{X}, \quad \text{for} \quad \varepsilon_X \sim \CN \left(0, \begin{pmatrix} 1 & 0 \\ 0 & s_e^2 \end{pmatrix} \right), \ \delta_e \sim \CN\left( \begin{pmatrix} \mu_e \\ 0 \end{pmatrix}, \Id_2 \right),
$$
with three environments: $P_1: \mu_e = 0, s_e = 1$, $P_2: \mu_e = 0, s_e = 3$, $P_3: \mu_e = 4, s_e = 1$, and $\BP_\mathrm{train} = 0.8 P_1 + 0.1 P_2 + 0.1 P_3$.
\end{example}
The left panel of \cref{fig:example_2_3} reports the population risks of CVaR anchor regression models including an intercept for varying $\gamma$ and $\alpha$ on $P_1$ and $P_3$.
Similarly to \cref{ex:1}, both decreasing $\alpha$ and increasing $\gamma$ can reduce the risk at $P_3$ but increases it at $P_1$, but simultaneously decreasing $\gamma$ and $\alpha$ can strictly improve performance on both environments.
Group CVaR performance is shown for $\alpha \in [0.001, 1]$. This performs worse than anchor regression as it puts more weight on the high-noise environment $P_2$. The kink in the group CVaR curve is at $\alpha = 0.2$.

\begin{example}[Misspecified model]
    \label{ex:3}
Let
$$
X = \delta_e + \varepsilon_X, \qquad Y = 0.1 X^3 + \varepsilon_Y,
\qquad \varepsilon_X, \varepsilon_Y \sim \CN(0,1), \qquad \delta_e\sim\mathrm{Unif}(0,4)
$$
\end{example}
The right panel of \cref{fig:example_2_3} shows the population risks of CVaR anchor regression models including an intercept for varying $\gamma$ and $\alpha$ at $\delta=0$ and $\delta=4$.
Unline \cref{ex:1}, the risks shown are for the endpoints of the training shift range.
As such, the plot looks qualitatively different from the right panel of \cref{fig:example_1}.
For fixed $\alpha$, varying $\gamma$ trades off risks at the two shifts.
For fixed $\gamma$, decreasing $\alpha$ lowers the risk at both shifts up until a point.

\begin{figure}[H]
    \centering
    \includegraphics[width=\textwidth]{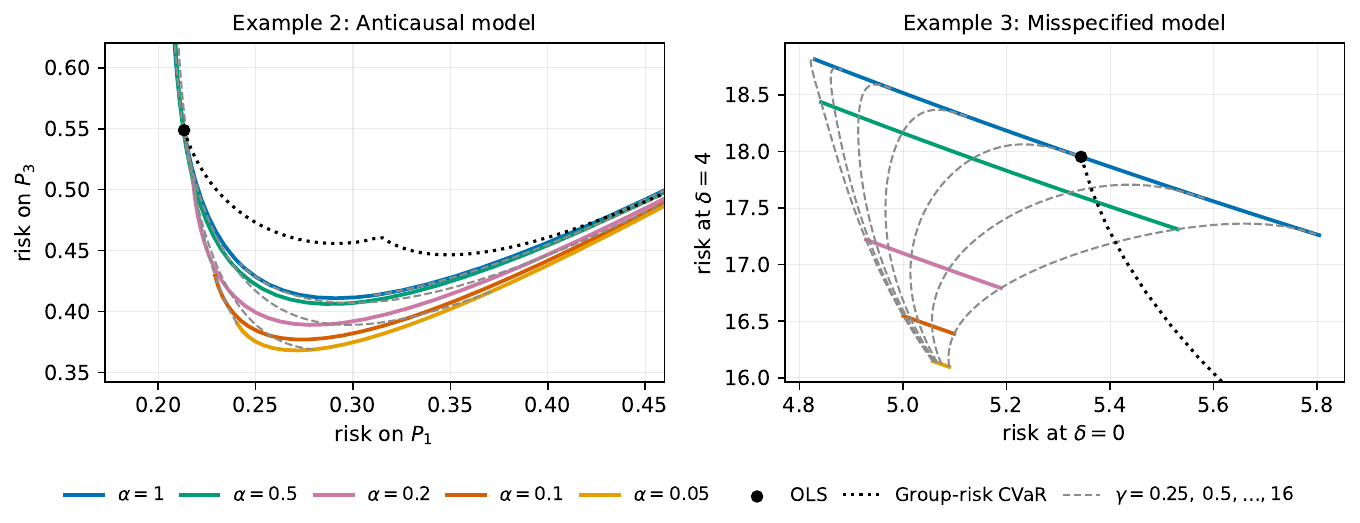}
    \caption{Population risks of CVaR anchor regression models for the anticausal model of \cref{ex:2} (left) and the misspecified model of \cref{ex:3} (right) for varying $\alpha$ and $\gamma$.}
    \label{fig:example_2_3}
\end{figure}

The previous examples illustrate the effect of $\alpha$ and $\gamma$ on the population risks of CVaR anchor regression.
The following shows the estimators increased variance for decreasing $\alpha$.
\begin{example}
\label{ex:4}
Let
$$
X=A+H + \varepsilon_X, \quad Y=\beta X+H+\varepsilon_Y, \quad H, A,\varepsilon_X,\varepsilon_Y, \overset{\text{i.i.d.}}{\sim} \CN(0,1).
$$
For $Z_i \overset{\text{i.i.d.}}{\sim} \CN(0,1)$ and $q_\alpha = F_{\chi^2_1}^{-1}(1 - \alpha)$, let
$$
c_\alpha = \mathrm{CVaR}_\alpha(Z_1^2) \quad \text{and} \quad 
v_\alpha = \Var\left((Z^2-q_\alpha)_+/\alpha\right) = \lim_{n\to\infty} n \Var(\widehat{\mathrm{CVaR}}_\alpha(Z_1^2,\ldots,Z_n^2)).
$$
Then
\begin{align*}
&\sqrt{n} \left( \hat b_\mathrm{cvar-anchor}^{\gamma,\alpha} - b_\mathrm{cvar-anchor}^{\gamma,\alpha} \right) \overset{d}{\to} \CN(0, V^{\gamma,\alpha}), \quad \text{where} \\
b_\mathrm{cvar-anchor}^{\gamma,\alpha} = &\beta + \frac{1}{2 + \gamma c_\alpha} \quad \text{and} \quad 
V^{\gamma,\alpha} =
2\left(\frac{\gamma c_\alpha}{2+\gamma c_\alpha}\right)^3
+\frac{3}{(2+\gamma c_\alpha)^2}
+\frac{\gamma^2(10c_\alpha^2+v_\alpha)}
{(2+\gamma c_\alpha)^4}.
\end{align*}
\end{example}
Here $b_\mathrm{cvar-anchor}^{\gamma,\alpha} = b_\mathrm{cvar-anchor}^{\gamma c_\alpha, 1}$ as predicted by \cref{cor:elliptical}.
\Cref{fig:example_4} shows the asymptotic variance as $\alpha$ is decreased, keeping $\gamma$ fixed (left) or varying $\gamma'(\alpha) = \gamma / c_\alpha$ to keep the effective penalty $\gamma c_\alpha$ constant (right).
For $\gamma$ fixed, the variance sharply increases as $\alpha$ decreases from $1$ to $0.1$.
In contrast, as long as $\alpha \geq 0.1$, the variance remains approximately constant when $\gamma \propto 1 / c_\alpha$.
\begin{figure}[H]
    \centering
    \includegraphics[width=0.9\textwidth]{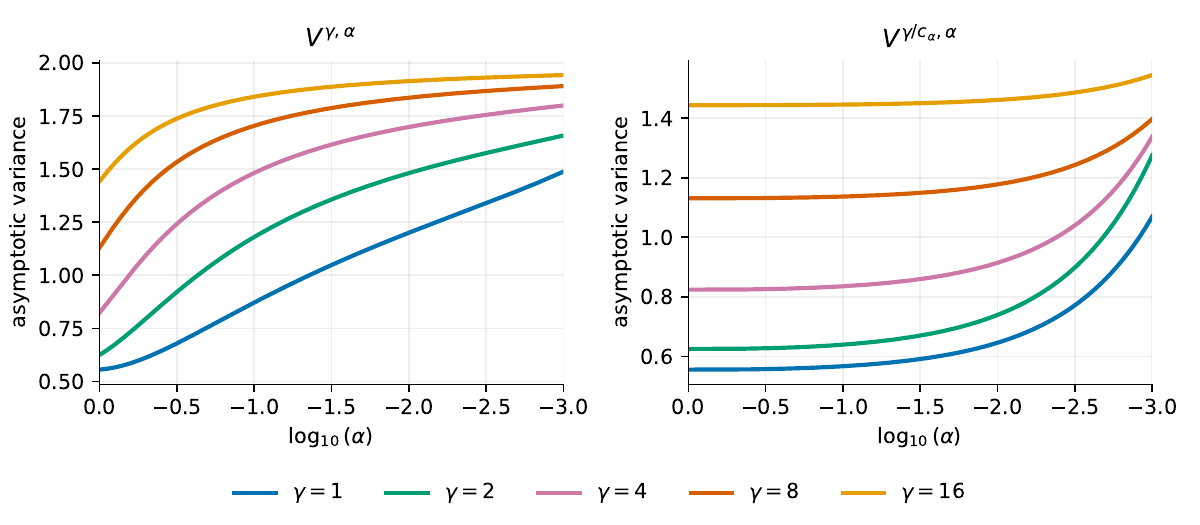}
    \caption{Asymptotic variance of CVaR anchor regression for the model of \cref{ex:4} as $\alpha$ decreases, with fixed $\gamma$ (left) or fixed effective penalty $\gamma c_\alpha$ (right).}
    \label{fig:example_4}
\end{figure}

For elliptical shifts, decreasing $\alpha$ while keeping $\gamma c_\alpha$ fixed leads to an estimator with the same population risk but increased variance.
That is, standard anchor regression dominates CVaR anchor regression.
However, as we have seen in \cref{ex:1,ex:2,ex:3}, if shifts are non-elliptical or the model is misspecified, lower values of $\alpha$ lead to a lower population risk.
CVaR anchor regression therefore presents a trade-off between better population risk and increased variance.
How many samples are needed for the lower population risk to outweigh the increased variance?
\Cref{fig:example_5} shows the empirical distribution of population test risks of finite-sample estimators as the training sample size increases.
For the data generating process of \cref{ex:1}, we use four environments with sample proportions $0.45, 0.45, 0.05, 0.05$ and increase $n$ from $100$ to $10'000$.
For the data generating processes of \cref{ex:2,ex:3}, we fix $n_e=10$ and increase the number of environments $| \CE |$ from $10$ to $1'000$.
We show the median and 10\%--90\% quantiles of the population test risks over $500$ simulations.
On the left row, we fix $\gamma=1$ and vary $\alpha$.
On the right, we vary $\gamma = c_{0.1} / c_\alpha \approx 4.4 / c_\alpha$ to keep the effective penalty $\gamma c_\alpha$ constant (Using $c_\alpha$ of the $\chi^2_1$ distribution).

For all three examples and with both $\gamma=1$ and $\gamma \propto c_\alpha^{-1}$, once $n \geq 1000$, both the median and $90\%$ percentiles of the risk of the CVaR anchor regression estimator decrease with decreasing $\alpha$.
The data generating process of \cref{ex:2} is the only one where smaller $\alpha$ potentially leads to a higher risk for small $n < 1000$.

\begin{figure}[H]
    \centering
    \includegraphics[width=0.95\textwidth]{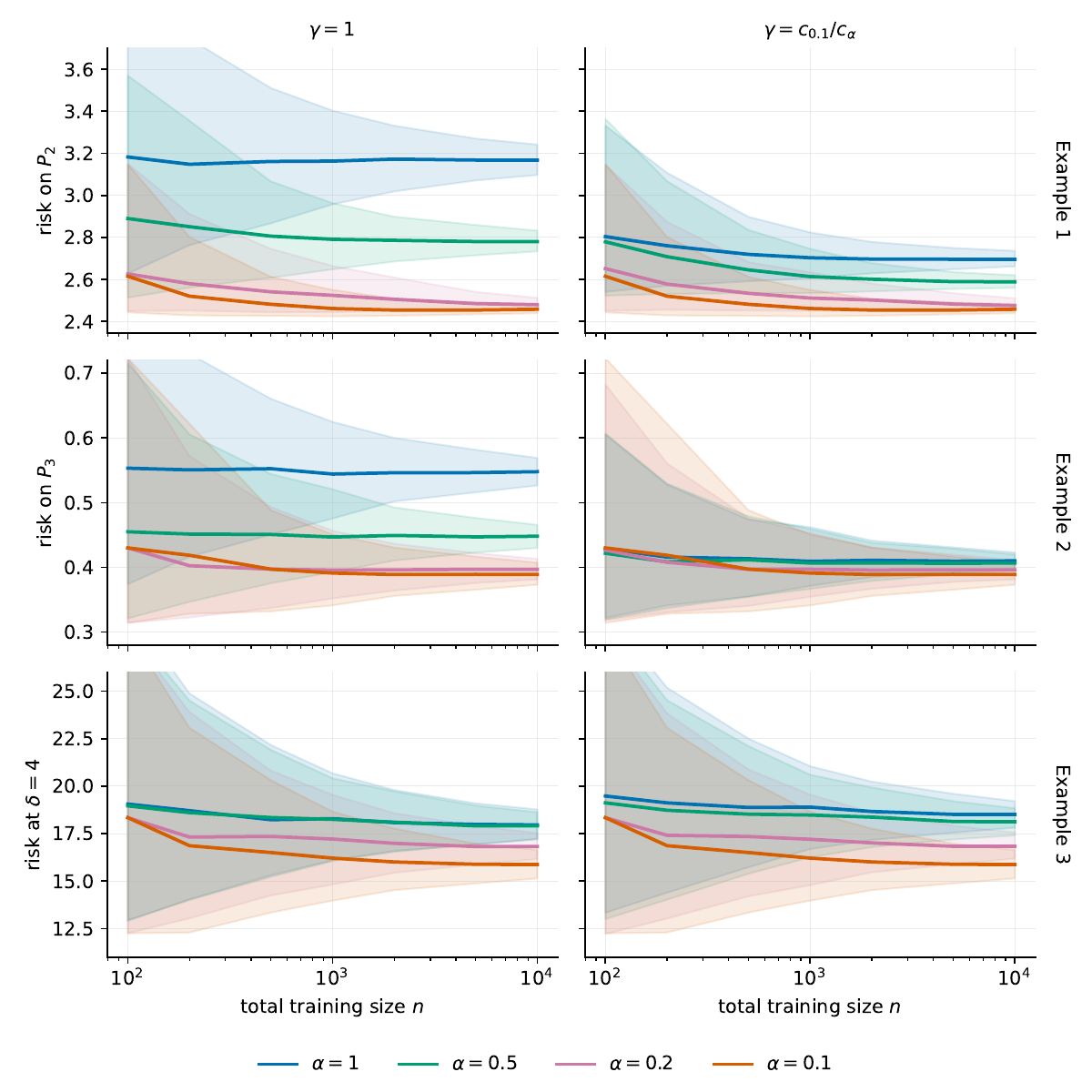}
    \caption{Population test risks for \cref{ex:1,ex:2,ex:3} (rows), with $\gamma=1$ (left) or $\gamma=c_{0.1}/c_\alpha$ (right), where $c_\alpha=\mathrm{CVaR}_\alpha(\chi^2_1)$. Lines show medians and shaded bands the 10th--90th percentiles over 500 simulations. Example~1 has four environments with sample proportions $0.45,0.45,0.05,0.05$. Examples~2 and~3 have $n_e=10$ and a growing number of environments; the environment-type proportions in Example~2 are fixed at $0.8,0.1,0.1$. Test risks are evaluated analytically.}
    \label{fig:example_5}
\end{figure}

\section{Nonlinear models}
\label{sec:nonlinear}
The anchor loss can be used with a nonlinear model $f_\theta$ by replacing $Y - X^T b$ with $r_\theta := Y - f_\theta(X)$.
When the resulting optimization problem is non-convex and algorithms only find local minima, one needs to choose between the worst-case reweighting formulation and the Rockafellar--Uryasev formulation.
The following extensions concern the estimation objective.
The shift-robustness guarantees of \cref{thm:cvar-anchor-robustness} requires a separate analysis.

\subsection{Tree-based models}
\label{sec:trees}
For gradient boosting models, as in anchor boosting \citep{buhlmann2020invariance,londschien2025domain}, worst-case reweighting appears more natural, updating the weights $q_e$ in each boosting iteration.
This has been implemented in \url{github.com/mlondschien/anchorboosting}.
We apply this to the NYC taxi dataset in \cref{sec:applications}.

\subsection{Deep learning}
\label{sec:deep-learning}
For neural networks, the Rockafellar--Uryasev formulation permits joint optimization of $t$ and $\theta$ by gradient descent, possibly replacing the hinge function with a smooth softplus approximation \citep{soma2020statistical}.
Alternatively, one can alternate network updates with updates of the environment weights \citep[as in GroupDRO,][]{sagawa2020distributionally}.
Both approaches require estimating the squared projected residuals during training.
For categorical anchors, these are the environment-wise squared mean residuals, whose estimation can be inaccurate when a minibatch contains only few observations from an environment.
This can be mitigated by constructing minibatches with a prescribed number of observations from each environment \citep[as in][Algorithm 1]{rame2022fishr} instead of drawing uniformly from the pooled data.
At fixed batch size, this creates a tradeoff between accurately estimating the environment-wise mean residuals and covering enough environments in each batch to estimate the CVaR tail.
Larger batches help.
The bias of the squared projected residuals can be corrected using the debiasing as described in \cref{sec:finite-sample-bias}.
However, applying a CVaR reweighting or hinge function to these noisy estimates introduces further bias, which the correction does not eliminate.
Tracking environment-wise mean residuals across iterations is another option \citep[as in][]{wang2014stochastic}, although these estimates become outdated as the network parameters change.
We apply the smoothed Rockafellar–Uryasev formulation with stratified minibatches and variance correction to the NYC taxi dataset in \cref{sec:applications}.

\subsection{Distributional CVaR anchor regression (classification)}
Beyond the squared loss, \citet{kook2022distributional} extend anchor regression to ordered and censored responses by replacing least-squares residuals with score residuals.
Applying $\cvar_\alpha$ to that penalty would give a distributional analogue of our estimator.

\section{NYC taxi dataset application}

\label{sec:applications}
We predict $Y = \log(\text{total amount})$ of a New York City yellow taxi trip from its log distance, log duration, their squares and product, the passenger count, two harmonics of the pickup hour, and dummies for weekday, rate code, and payment type.
After filtering, we retain $6'033'360$ of the $6'310'419$ records in the July 2019 yellow taxi file.
We use the discrete pickup zone as the anchor, giving $|\CE| = 261$ environments.
We fit linear CVaR anchor regression estimators on all retained July trips and evaluate on a fixed test set of $202'790$ trips from 235 zones in April 2020 during the COVID-19 pandemic.

In addition to the linear model, we fit CVaR anchor boosting models on the same July data, with up to $1000$ trees, maximum depth $3$, and learning rate $0.1$, as described in \cref{sec:nonlinear}.
Again, simultaneously decreasing $\alpha$ and $\gamma$ lowers the $\cvar_{0.1}$ at similar mean risk.

Finally, we fit a feedforward neural network with 3 hidden layers of 64 units each, ReLU activations, and a linear output layer (see \cref{sec:deep-learning}).
We reserve approximately $10\%$ of the July trips within each zone for validation, leaving $5'430'013$ training trips and $603'347$ validation trips.
We use Adam with learning rate 0.001 and weight decay 0.0001 on the network parameters.
We first train with 5 warmup epochs using squared error loss, followed by 25 epochs with the CVaR anchor regression loss.
We select the checkpoint with the lowest validation $\cvar_{0.1}$ of the zone-wise risks among these 25 epochs.
For $\alpha<1$, we use the Rockafellar--Uryasev formulation with $t$ as an additional parameter without weight decay and the hinge function replaced by a softplus approximation with temperature $\tau = 0.001$.
For each batch, we sample 32 observations with replacement from each of the 261 training zones, resulting in a batch size of $8'352$.
Each epoch comprises 651 such batches.
We debias the squared projected residuals as described in \cref{sec:finite-sample-bias}.
The loss is reweighted by the relative zone counts.

The left column of \Cref{fig:taxi} shows weighted mean and $\cvar_{0.1}$ of the zone-wise risks for varying $\gamma$ and $\alpha$.
Simultaneously decreasing $\alpha$ and $\gamma$ lowers the $\cvar_{0.1}$ at similar mean risk.
The right column of \Cref{fig:taxi} shows ranked zone-wise risks for fits with similar mean risk.
Decreasing $\alpha$ not only improves the $\cvar_{0.1}$, but also $\alpha$-quantiles for $\alpha \leq 0.2$.

\begin{figure}[H]
  \centering
  \includegraphics[width=\textwidth]{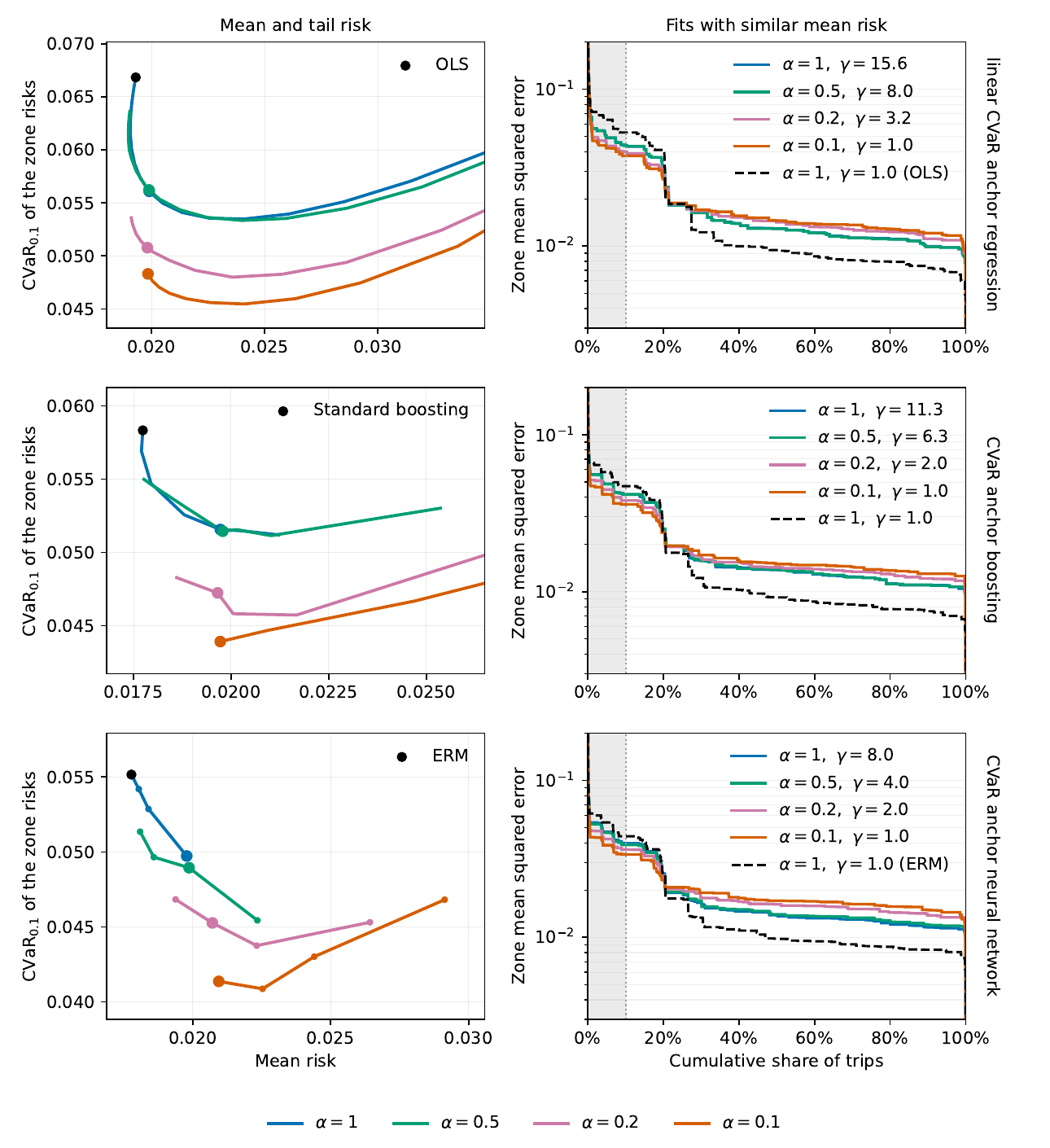}
  \caption{
    \label{fig:taxi}Zone-wise test risks in April 2020.
  Rows: linear CVaR anchor regression, CVaR anchor boosting, and CVaR anchor neural network. Left: $\cvar_{0.1}$ against mean zone risk for $\alpha\in\{1,0.5,0.2,0.1\}$ as $\gamma \geq 1$ increases.
  Large coloured points highlight similar April mean risks within each row.
  Right: their zone risks, ranked separately and weighted by trip counts.
  Neural-network results are medians over three seeds.}
\end{figure}

\section{Discussion and future work}

\subsection{Other divergence penalties}
CVaR admists a DRO formulation with $f$-divergence generated by $f_\alpha(t)=0$ for $0\leq t\leq 1/\alpha$ and $f_\alpha(t)=\infty$ otherwise.
A natural extension of CVaR anchor regression is to replace $\cvar_\alpha((P_A(Y-X^T b))^2)$ by $\sup_{\BQ \ll \BP_\mathrm{train}:\, d_f(\BQ\|\BP_\mathrm{train})\leq\rho} \BE_\BQ[(P_A(Y-X^T b))^2]$, where $d_f$ is another $f$-divergence and $\rho$ controls the size of the robustness set.
This preserves convexity of the linear estimator and \cref{thm:cvar-anchor-robustness} transfers directly.
A natural candidate is the $\chi^2$-divergence, with $f(t)=(t-1)^2/2$.
For $\rho$ small enough, the penalty reduces to $\BE[(P_A(Y-X^T b))^2] + \sqrt{2 \rho \Var((P_A(Y-X^T b))^2)}$.
We focus on CVaR here due to its interpretability as a tail average.

\subsection{Another parametrization}
Our formulation of CVaR anchor regression extends standard anchor regression by replacing the expectation of the squared projected residuals with its CVaR.
Alternatively, for $\gamma \geq 1$, the anchor regression loss can be seen as a penalized mean squared error loss.
Replacing the expectation in that penalty with a CVaR yields
$$
\BE[(Y-X^T b)^2]  +(\gamma-1)\cvar_\alpha\big((P_A(Y-X^T b))^2\big).
$$
This corresponds to reweightings $\{ \BQ \mid \BQ \text{ is a prob. measure with } 1/\gamma \leq \BQ / \BP_\mathrm{train} \leq (1 + \gamma + \alpha) / (\gamma \alpha) \}$.

\subsection{Continuous anchors}
Although the robustness guarantees of \cref{thm:cvar-anchor-robustness} allow continuous anchors, \cref{ex:1,ex:2,ex:3} and the NYC taxi dataset application use discrete anchors.
This is in line with the motivation of the introduction: Many environments, only a few of which exibit large shifts.
Future work could evaluate the method on settings where continuous measurements are used as anchors and investigate when CVaR anchor regression improves upon standard anchor.

\subsection{Projections with respect to worst-case reweightings}
The (CVaR) anchor regression objective uses the projection $P_A Z = \BE[Z A^T] \BE[A A^T]^{\dagger} A$ under $\BP_\mathrm{train}$, even when evaluating the penalty under $\BQ$.
Write $P_A^\BQ Z = \BE_\BQ[Z A^T] \BE_\BQ[A A^T]^{\dagger} A$ for the projection under $\BQ$.
One could also consider the objective
$$
\BE[(M_A(Y-X^T b))^2]
+ \gamma \sup_{\BQ\in\CQ_\alpha} \BE_\BQ[(P_A^\BQ (Y-X^T b))^2],
$$
where we additionally restrict the densities $\dd \BQ / \dd \BP_\mathrm{train} = w(A)$ to be $A$-measurable, with $0 \leq w \leq 1/\alpha$ and $\BE[w(A)] = 1$.
For anchors given by environment indicators, this yields the same objective as CVaR anchor regression, but for continuous anchors the projection can change under reweighting.
However, in that case it is not clear how to restrict the reweightings to be $A$-measurable with finite samples.
If we do not restrict $\BQ$ to $A$-measurable reweightings, the reweighting $\BQ$ can change the distribution of observations within environments, including the distribution of the noise.
Thus $\BE[\varepsilon | A ] = 0 \not \Rightarrow \BE_\BQ[\varepsilon | A] = 0$.
The anchor penalty could be positive even if there are no differences between environments and the resulting objective would be closer to observation-level CVaR DRO \citep{duchi2016statistics}.

\section*{Declaration of AI use}
Claude Code and OpenAI Codex were used throughout the preparation of this manuscript and the accompanying software.
Their use included assistance with writing, structure, and language, identifying errors, formalizing and improving proofs, and in particular developing the statements and proofs of \cref{prop:conditional-risks,prop:intercepts}.
They also assisted in formulating \cref{ex:1,ex:2,ex:3}, calculations in \cref{ex:4}, computing population risks, generating \cref{fig:example_1,fig:example_2_3,fig:example_4,fig:example_5,fig:taxi}, and implementing CVaR anchor regression in \url{https://github.com/mlondschien/cvar-anchor-regression}.
The author takes responsibility for the final manuscript, code, and reported results.
In particlar, all AI-generated proofs were checked for correctness and rewritten to improve the exposition.

\bibliography{references}

\begin{thebibliography}{}

\bibitem[\protect\citeauthoryear{Agarwal and Zhang}{Agarwal and
  Zhang}{2022}]{agarwal2022minimax}
Agarwal, A. and T.~Zhang (2022).
\newblock Minimax regret optimization for robust machine learning under
  distribution shift.
\newblock In {\em Proceedings of the 35th Conference on Learning Theory},
  Volume 178 of {\em Proceedings of Machine Learning Research}, pp.\
  2704--2729. PMLR.

\bibitem[\protect\citeauthoryear{Ben-Tal, den Hertog, De~Waegenaere, Melenberg,
  and Rennen}{Ben-Tal et~al.}{2013}]{bental2013robust}
Ben-Tal, A., D.~den Hertog, A.~De~Waegenaere, B.~Melenberg, and G.~Rennen
  (2013).
\newblock Robust solutions of optimization problems affected by uncertain
  probabilities.
\newblock {\em Management Science\/}~{\em 59\/}(2), 341--357.

\bibitem[\protect\citeauthoryear{B{\"u}hlmann}{B{\"u}hlmann}{2020}]{buhlmann2020invariance}
B{\"u}hlmann, P. (2020).
\newblock Invariance, causality and robustness.
\newblock {\em Statistical Science\/}~{\em 35\/}(3), 404--426.

\bibitem[\protect\citeauthoryear{Duchi, Glynn, and Namkoong}{Duchi
  et~al.}{2021}]{duchi2016statistics}
Duchi, J.~C., P.~W. Glynn, and H.~Namkoong (2021).
\newblock Statistics of robust optimization: A generalized empirical likelihood
  approach.
\newblock {\em Mathematics of Operations Research\/}~{\em 46\/}(3), 946--969.

\bibitem[\protect\citeauthoryear{Duchi and Namkoong}{Duchi and
  Namkoong}{2021}]{duchi2021learning}
Duchi, J.~C. and H.~Namkoong (2021).
\newblock Learning models with uniform performance via distributionally robust
  optimization.
\newblock {\em The Annals of Statistics\/}~{\em 49\/}(3), 1378--1406.

\bibitem[\protect\citeauthoryear{Eastwood, Robey, Singh, von K{\"u}gelgen,
  Hassani, Pappas, and Sch{\"o}lkopf}{Eastwood
  et~al.}{2022}]{eastwood2022probable}
Eastwood, C., A.~Robey, S.~Singh, J.~von K{\"u}gelgen, H.~Hassani, G.~J.
  Pappas, and B.~Sch{\"o}lkopf (2022).
\newblock Probable domain generalization via quantile risk minimization.
\newblock In {\em Advances in Neural Information Processing Systems},
  Volume~35, pp.\  17340--17358. Curran Associates, Inc.

\bibitem[\protect\citeauthoryear{Hjort and Pollard}{Hjort and
  Pollard}{1993}]{hjort1993asymptotics}
Hjort, N.~L. and D.~Pollard (1993).
\newblock Asymptotics for minimisers of convex processes.
\newblock Unpublished manuscript.

\bibitem[\protect\citeauthoryear{Koh, Sagawa, Marklund, Xie, Zhang,
  Balsubramani, Hu, Yasunaga, Phillips, Gao, Lee, David, Stavness, Guo,
  Earnshaw, Haque, Beery, Leskovec, Kundaje, Pierson, Levine, Finn, and
  Liang}{Koh et~al.}{2021}]{koh2021wilds}
Koh, P.~W., S.~Sagawa, H.~Marklund, S.~M. Xie, M.~Zhang, A.~Balsubramani,
  W.~Hu, M.~Yasunaga, R.~L. Phillips, I.~Gao, T.~Lee, E.~David, I.~Stavness,
  W.~Guo, B.~Earnshaw, I.~Haque, S.~M. Beery, J.~Leskovec, A.~Kundaje,
  E.~Pierson, S.~Levine, C.~Finn, and P.~Liang (2021).
\newblock {WILDS}: A benchmark of in-the-wild distribution shifts.
\newblock In {\em Proceedings of the 38th International Conference on Machine
  Learning}, Volume 139 of {\em Proceedings of Machine Learning Research}, pp.\
   5637--5664. PMLR.

\bibitem[\protect\citeauthoryear{Kook, Sick, and B{\"u}hlmann}{Kook
  et~al.}{2022}]{kook2022distributional}
Kook, L., B.~Sick, and P.~B{\"u}hlmann (2022).
\newblock Distributional anchor regression.
\newblock {\em Statistics and Computing\/}~{\em 32}, 39.

\bibitem[\protect\citeauthoryear{Krueger, Caballero, Jacobsen, Zhang, Binas,
  Zhang, Le~Priol, and Courville}{Krueger et~al.}{2021}]{krueger2021out}
Krueger, D., E.~Caballero, J.-H. Jacobsen, A.~Zhang, J.~Binas, D.~Zhang,
  R.~Le~Priol, and A.~Courville (2021).
\newblock Out-of-distribution generalization via risk extrapolation ({REx}).
\newblock In {\em Proceedings of the 38th International Conference on Machine
  Learning}, Volume 139 of {\em Proceedings of Machine Learning Research}, pp.\
   5815--5826. PMLR.

\bibitem[\protect\citeauthoryear{Londschien, Burger, R{\"a}tsch, and
  B{\"u}hlmann}{Londschien et~al.}{2026}]{londschien2025domain}
Londschien, M., M.~Burger, G.~R{\"a}tsch, and P.~B{\"u}hlmann (2026).
\newblock Domain generalization and adaptation in intensive care with anchor
  regression.
\newblock {\em {RSS}: Data Science and Artificial Intelligence\/}~{\em 2\/}(1),
  udag001.

\bibitem[\protect\citeauthoryear{Meinshausen and B{\"u}hlmann}{Meinshausen and
  B{\"u}hlmann}{2015}]{meinshausen2015maximin}
Meinshausen, N. and P.~B{\"u}hlmann (2015).
\newblock Maximin effects in inhomogeneous large-scale data.
\newblock {\em The Annals of Statistics\/}~{\em 43\/}(4), 1801--1830.

\bibitem[\protect\citeauthoryear{Mo, Tang, Xue, Liu, and Zhu}{Mo
  et~al.}{2024}]{mo2024minimax}
Mo, W., W.~Tang, S.~Xue, Y.~Liu, and J.~Zhu (2024).
\newblock Minimax regret learning for data with heterogeneous subgroups.
\newblock arXiv:2405.01709.

\bibitem[\protect\citeauthoryear{Oren, Sagawa, Hashimoto, and Liang}{Oren
  et~al.}{2019}]{oren2019distributionally}
Oren, Y., S.~Sagawa, T.~B. Hashimoto, and P.~Liang (2019).
\newblock Distributionally robust language modeling.
\newblock In {\em Proceedings of the 2019 Conference on Empirical Methods in
  Natural Language Processing and the 9th International Joint Conference on
  Natural Language Processing ({EMNLP-IJCNLP})}, pp.\  4227--4237. Association
  for Computational Linguistics.

\bibitem[\protect\citeauthoryear{Ram{\'e}, Dancette, and Cord}{Ram{\'e}
  et~al.}{2022}]{rame2022fishr}
Ram{\'e}, A., C.~Dancette, and M.~Cord (2022).
\newblock {Fishr}: Invariant gradient variances for out-of-distribution
  generalization.
\newblock In {\em Proceedings of the 39th International Conference on Machine
  Learning}, Volume 162 of {\em Proceedings of Machine Learning Research}, pp.\
   18347--18377. PMLR.

\bibitem[\protect\citeauthoryear{Rockafellar and Uryasev}{Rockafellar and
  Uryasev}{2000}]{rockafellar2000optimization}
Rockafellar, R.~T. and S.~Uryasev (2000).
\newblock Optimization of conditional value-at-risk.
\newblock {\em Journal of Risk\/}~{\em 2\/}(3), 21--41.

\bibitem[\protect\citeauthoryear{Rothenh{\"a}usler, Meinshausen, B{\"u}hlmann,
  and Peters}{Rothenh{\"a}usler et~al.}{2021}]{rothenhausler2021anchor}
Rothenh{\"a}usler, D., N.~Meinshausen, P.~B{\"u}hlmann, and J.~Peters (2021).
\newblock Anchor regression: Heterogeneous data meet causality.
\newblock {\em Journal of the Royal Statistical Society Series B: Statistical
  Methodology\/}~{\em 83\/}(2), 215--246.

\bibitem[\protect\citeauthoryear{Sagawa, Koh, Hashimoto, and Liang}{Sagawa
  et~al.}{2020}]{sagawa2020distributionally}
Sagawa, S., P.~W. Koh, T.~B. Hashimoto, and P.~Liang (2020).
\newblock Distributionally robust neural networks for group shifts: On the
  importance of regularization for worst-case generalization.
\newblock In {\em International Conference on Learning Representations}.

\bibitem[\protect\citeauthoryear{Soma and Yoshida}{Soma and
  Yoshida}{2020}]{soma2020statistical}
Soma, T. and Y.~Yoshida (2020).
\newblock Statistical learning with conditional value at risk.
\newblock arXiv:2002.05826.

\bibitem[\protect\citeauthoryear{Wang, Fang, and Liu}{Wang
  et~al.}{2017}]{wang2014stochastic}
Wang, M., E.~X. Fang, and H.~Liu (2017).
\newblock Stochastic compositional gradient descent: Algorithms for minimizing
  compositions of expected-value functions.
\newblock {\em Mathematical Programming\/}~{\em 161\/}(1--2), 419--449.

\bibitem[\protect\citeauthoryear{Wang, B{\"u}hlmann, and Guo}{Wang
  et~al.}{2026}]{wang2023distributionally}
Wang, Z., P.~B{\"u}hlmann, and Z.~Guo (2026).
\newblock Distributionally robust learning for multisource unsupervised domain
  adaptation.
\newblock {\em The Annals of Statistics\/}~{\em 54\/}(2), 570--596.

\end{thebibliography}

\appendix

\section{Details on including an intercept}
\label{sec:intercepts}
\paragraph{Why do the approaches (i) and (ii) differ for CVaR anchor regression with $\alpha<1$?}
Recall that (i): center $X$, $A$, and $Y$, run the regression without an intercept, and compute the intercept as $\BE[Y-X^Tb]$ and (ii): include a constant column in both $X$ and $A$.

For random $Z$, denote with $Z_c = Z - \BE[Z]$ its centered version and $\bar Z = \BE[Z]$ its mean.
Then $M_{A, 1} Z = M_{A_c, 1} Z = M_{A_c} Z - P_1 Z = M_{A_c} Z_c$ and $P_{A, 1} Z = P_{(A_c, 1)} Z = P_{A_c} Z + P_1 Z = P_{A_c} Z_c + \bar Z$.
For anchor regression with $\alpha=1$, the equivalence follows from the decomposition
\begin{multline*}
\BE[ ( M_{A, 1} (Y - X^T b - b_0) )^2 ] + \gamma \BE[ (P_{A, 1} (Y - X^T b - b_0))^2 ] = \BE[ (M_{A_c} (Y_c - X_c^T b))^2 ] \\ + \gamma ( \BE[ (P_{A_c} (Y_c - X_c^T b))^2 ] + 2 \underbrace{\BE[ P_{A_c} (Y_c - X_c^T b) P_1 (Y - X^T b - b_0) ]}_{= \BE[ P_{A_c} (Y_c - X_c^T b) ] (\bar Y - \bar X^T b - b_0) = 0} + (\bar y - \bar X^T b - b_0)^2 )
\end{multline*}
Thus, with $b_0^\star(b) = \bar Y - \bar X^T b$ the losses equal and so does the minimizer $b^\star$.

For CVaR anchor regression with $\alpha<1$, we replace the $\BE$ terms in the anchor penalty with a supremum over $\BE_\BQ$ for $\BQ \in \CQ_\alpha$.
The term $\BE_\BQ[ P_{A_c} (Y_c - X_c^T b) ]$ needs no longer be zero and so the intercept can no longer be extracted from the penalty.

Consider the following example: $\mX = 0$, $(Y, A) = (0,0)$ with probability $3/4$ and $(Y, A) = (1,1)$ with probability $1/4$.
Let $\alpha = 1/4$.
For (i), we get $b_0^\star = \BE[Y] = \bar Y = 1/4$ the residual mean with respect to the training distribution.
For (ii), we get $\cvar_\alpha(P_A(Y - b_0))^2 = \max\{ (1 - b_0)^2, b_0^2 \}$ and so $b_0^\star = 1/2$.
The minimizing intercept equals the residual mean at the worst-case $\BQ \in \CQ_\alpha$.

\paragraph{What are the robustness sets for the two approaches?}
For (i), we simply subtract all variables mean under the training distribution and the robustness set is
$$
\bigcup_{\BQ \in \CQ_\alpha} \{ \BE[MA] + v : v \indep \varepsilon, \BE[vv^T] \preceq \gamma M \BE_\BQ[ (A - \BE[A]) (A - \BE[A])^T ]  M^T \}.
$$
As the offset is with respect to the training distribution, $\BE_\BQ[ (A - \BE[A]) (A - \BE[A])^T ] \neq \Cov_\BQ(A)$.

For (ii), it turns out that for $\gamma \neq 1$ there is no robustness $C$ set satisfying
$$
\sup_{v \in C} \BE_v[ (Y - X^T b - b_0)^2 ] = \BE[ (M_{A, 1} (Y - X^T b - b_0))^2 ] + \gamma \cvar_\alpha( (P_{A, 1} (Y - X^T b - b_0))^2 ).
$$
To see this, note that for large $b_0$, the right hand side scales with $\gamma b_0^2$ while the left hand side scales with $b_0^2$.
However, as the following proposition shows, the minimizing $b, b_0$ have an interpretation as worst-case robust regression coefficients with respect to a robustness set $C_2$.
\begin{proposition}
\label{prop:intercepts}
Assume the structural model of \cref{sec:structural-model} and that $P_A 1=1$.
Fix $\alpha\in(0,1]$, $\gamma\geq0$, and define
\begin{align*}
\CQ_\alpha&=\left\{\BQ\ll\BP_{\rm train}:\ \BQ\text{ is a probability measure},\quad
\left\|\frac{\dd\BQ}{\dd\BP_{\rm train}}\right\|_\infty\leq\frac1\alpha\right\},\\
C_1&=\bigcup_{\BQ\in\CQ_\alpha}
\left\{\BE_\BQ[MA]+v:\ v\indep\varepsilon,\quad \BE[vv^T]\preceq\gamma M\Cov_\BQ(A)M^T\right\} \quad \text{and}
\\
C_2&=\bigcup_{\BQ\in\CQ_\alpha}
\left\{\BE_\BQ[MA]+v:\ v\indep\varepsilon,\quad \BE[v]=0, \quad \BE[vv^T]\preceq\gamma M\Cov_\BQ(A)M^T\right\}.
\end{align*}
Note how $C_1$ and $C_2$ differ by the additional condition $\BE[v]=0$ in $C_2$.
Write $r:=Y-X^Tb-b_0$. Then 
\begin{align}
\sup_{w\in C_1}\BE_w[r^2]
&=\BE[(M_A r)^2]+\sup_{\BQ\in\CQ_\alpha}
\left(\left|\BE_\BQ[P_A r]\right|
+\sqrt{\gamma\Var_\BQ(P_A r)}\right)^2 \quad \text{and} \label{eq:intercepts-1}
\\
\sup_{w\in C_2}\BE_w[r^2]
&=\BE[(M_A r)^2]+\sup_{\BQ\in\CQ_\alpha}
\left\{\bigl(\BE_\BQ[P_A r]\bigr)^2
+\gamma\Var_\BQ(P_A r)\right\} \label{eq:intercepts-2}
\end{align}
Consider additionally the CVaR anchor loss with intercept by including a constant column into both $X$ and $A$ (such that $P_A 1=1$)
\begin{equation}
\label{eq:intercepts-3}
\BE[(M_A (Y - X^T b - b_0))^2]+\gamma\cvar_\alpha((P_A (Y - X^T b - b_0))^2).
\end{equation}
Then,
\begin{itemize}[(i)]
    \item For $\alpha=1$ and $\gamma>0$ the three losses \eqref{eq:intercepts-1}, \eqref{eq:intercepts-2}, and \eqref{eq:intercepts-3} have the same minimizers $(b_0, b)$ and the same minimum value.
    \item For any $\alpha\in(0,1]$ and $\gamma\geq1$, the losses \eqref{eq:intercepts-2}, and \eqref{eq:intercepts-3} have the same minimizers $(b_0, b)$ and the same minimum value.
\end{itemize}
\end{proposition}

\section{Proofs}
\label{sec:proofs}
\subsection{\texorpdfstring{Proof of \Cref{thm:anchor-regression} and \Cref{thm:cvar-anchor-robustness}}{Proof of Theorems \ref{thm:anchor-regression} and \ref{thm:cvar-anchor-robustness}}} \label{app:proof-cvar-anchor-robustness}
We prove \Cref{thm:cvar-anchor-robustness}, which implies \Cref{thm:anchor-regression} as a special case for $\alpha = 1$.
The proof below is for random $v$ with $\BE[v v^T] \preceq \gamma M \BE_\BQ[A A^T] M^T$ for some $\BQ \in \CQ_\alpha$.
\begin{proof}
Write $r_b := ((\Id - B)^{-1})^T (-b^T, 1, 0)^T$ so $Y - X^T b = r_b^T(\varepsilon + v)$ under $\BP_v$ and $Y - X^T b = r_b^T(\varepsilon + MA)$ under $\BP_\mathrm{train}$.
    As $\Cov(\varepsilon, A) = 0 \Rightarrow P_A \varepsilon = 0$, we have
    \begin{equation}
    \label{eq:thm-cvar-robustness-1}        
    P_A (Y - X^T b) = P_A r_b^T \varepsilon + P_A r_b^T M A = r_b^T M A \quad \text{and} \quad M_A (Y - X^T b) = r_b^T \varepsilon.
\end{equation}
    Under $\BP_v$, the shift $v$ is independent of $\varepsilon$ with $\BE[\varepsilon] = 0$, so $\BE_v[\varepsilon v] = 0$ and
    \begin{equation}
    \label{eq:thm-cvar-robustness-2}    
    \BE_v [ (Y - X^T b)^2] = \BE_v[(r_b^T \varepsilon)^2] + \BE_v[(r_b^T v)^2].
\end{equation}
    For $\BQ \in \CQ_\alpha$ and $v$ satisfying $\BE[v v^T] \preceq \gamma M \BE_\BQ[A A^T] M^T =: K \succeq 0$, we have
\begin{align}\label{eq:thm-cvar-robustness-3}
    \BE_v [ (r_b^T v)^2 ] = \BE_v [ r_b^T v v^T r_b ] \leq r_b^T K r_b
    = \gamma  \BE_\BQ[(P_A (Y - X^T b))^2].
\end{align}
Combining \cref{eq:thm-cvar-robustness-1,eq:thm-cvar-robustness-2,eq:thm-cvar-robustness-3} gives
\begin{align*}
    \BE_v [ (Y - X^T b)^2] = \BE_v[(r_b^T \varepsilon)^2] + \BE_v[(r_b^T v)^2] \leq \BE[ (M_A(Y - X^T b))^2] + \gamma \BE_\BQ[(P_A (Y - X^T b))^2]
\end{align*}
and taking the supremum over $\BQ \in \CQ_\alpha$ and $\{v \mid \BE_v[v v^T] \preceq \gamma M \BE_\BQ[A A^T] M^T\}$ gives
\begin{equation}
    \label{eq:thm-cvar-robustness-4}
    \sup_{v \in C_{\alpha, \gamma}} \BE_v [ (Y - X^T b)^2] \leq \BE[ (M_A(Y - X^T b))^2] + \gamma \cvar_\alpha((P_A (Y - X^T b))^2).
\end{equation}
Finally, equality in \cref{eq:thm-cvar-robustness-3} is attained for deterministic $v' = K r_b / (r_b^T K r_b)^{1/2}$ and so the inequality in \cref{eq:thm-cvar-robustness-4} is an equality.
\end{proof}

\subsection{\texorpdfstring{Proofs of \cref{sec:geometry}}{Proofs of Section \ref{sec:geometry}}}
We begin with the proof of \cref{prop:nesting}.
\begin{proof}
If $\gamma' = 0$, both sets are equal to $\{0\}$.
Otherwise, let $v \in C_{\alpha, \gamma}$, witnessed by $\BQ \in \CQ_\alpha$ with density $h := \dd\BQ / \dd\BP_\mathrm{train} \leq 1/\alpha$, and set $c := \gamma / \gamma' \leq 1$.
Then $0 \leq ch \leq 1/\alpha'$ as $\gamma/\alpha \leq \gamma' /\alpha'$ and $\BE[ch] = \BE_\BQ[c] =c.$
If $c=1$, set $\lambda=0$. Else, let $\lambda := (1 - c)/(1/\alpha' - c) \in [0, 1]$ and define $h' := ch + \lambda(1/\alpha' - ch)$.
Then $E[h'] = c + \lambda(1/\alpha' - c) = 1$, so $h'$ defines a probability measure $\BQ'.$
As it's a convex combination of $ch$ and $1 / \alpha'$, we have $\BQ' \in \CQ_{\alpha'}$.
As $h' \geq ch \Rightarrow h' \gamma' \geq \gamma h$, we have
$$
vv^T \preceq \gamma M \BE_\BQ[AA^T] M^T
\preceq \gamma' M \BE_{\BQ'}[AA^T] M^T.
$$
Hence $v \in C_{\alpha', \gamma'}$.
\end{proof}

Next, we prove \cref{prop:finite-convex-hull}.
\begin{proof}
Since $w_e/\alpha \geq 1$, every probability vector $(q_e)_{e \in \CE}$ is admissible, so
$$
C_{\alpha, \gamma} = \bigcup_{q_e \geq 0,\,\sum_e q_e = 1}
\Big\{v : vv^T \preceq \gamma\sum_e q_e\delta_e\delta_e^T\Big\}.
$$
This set is convex and contains $\pm\sqrt{\gamma}\delta_e$ and thus also their convex hull.
Conversely, for every $v \in C_{\alpha, \gamma}$, there is some $(q_e)_{e \in \CE}$ such that
  $vv^T \preceq \gamma \sum_e q_e \delta_e \delta_e^T$.
  Thus, for every $u \in \BR^{\mX + \mH + 1}$,
  $$
  (u^T v)^2
  \leq \gamma \sum_e q_e (u^T \delta_e)^2
  \leq \gamma \max_e |u^T \delta_e|^2.
  $$
  Taking square roots gives $|u^T v| \leq \sqrt{\gamma}\max_e |u^T \delta_e|$.
  This characterizes the convex hull of $\{\pm\sqrt{\gamma}\delta_e : e \in \CE\}$ and so $v$ lies in it.
\end{proof}

\subsection{\texorpdfstring{Proofs of \cref{sec:equivalences}}{Proofs of Section \ref{sec:equivalences}}}
We begin with the proof of \cref{lem:elliptical}.
\begin{proof}
Under \cref{eq:structural}, $P_A(Y-X^Tb) = r_b^T MA$ for $r_b := ((\Id-B)^{-1})^T(-b^T,1,0)^T$.
Using $u = r_b$, the assumption of elliptical shifts gives
$$
(P_A(Y-X^Tb))^2 \stackrel{d}{=} Z^2 u^T\BE[MAA^TM^T]u = Z^2 \BE[(P_A(Y-X^Tb))^2].
$$
The second conclusion follows using the homogeneity: $T(Z^2 \BE[(P_A(Y-X^Tb))^2]) = T(Z^2) \BE[(P_A(Y-X^Tb))^2]$.
\end{proof}

Next, we prove \cref{prop:elliptical}.
\begin{proof}
The identity follows from \cref{lem:elliptical}.
The bounds $1 \leq c_\alpha \leq 1/\alpha$ follow from $\BE[Z^2] \leq \cvar_\alpha(Z^2) \leq \BE[Z^2]/\alpha$ and $\BE[Z^2]=1$.
\end{proof}

Next, we prove \cref{prop:conditional-risks}.
\begin{proof}
Let $r_b:= ((\mathrm{Id}-B)^{-1})^T (-b^T, 1, 0)^T$ so that $P_A(Y-X^Tb) = r_b^T MA$, $P_A(\tilde Y - \tilde X^T b) = \sqrt{\gamma} r_b^T MA$, $M_A(Y-X^Tb) = r_b^T \varepsilon$, and $M_A(\tilde Y - \tilde X^T b) = r_b^T \tilde\varepsilon$.
Using independence between $\tilde \varepsilon$ and $A$, $\BE[\tilde\varepsilon]=\BE[\varepsilon]=0$, and the fact that $\varepsilon$ and $\tilde\varepsilon$ have the same marginal law, calculate
\begin{multline*}
\BE[(\tilde Y-\tilde X^Tb)^2\mid A] = \BE[(r_b^T \tilde\varepsilon + \sqrt{\gamma} r_b^T MA)^2\mid A] \\
= \BE[(r_b^T\tilde\varepsilon)^2] + \gamma (r_b^TMA)^2 = \BE[(M_A(Y-X^Tb))^2] + \gamma (P_A(Y-X^Tb))^2.
\end{multline*}
Taking the CVaR on both sides and using homogeneity and translation equivariance proves the identity.
\end{proof}

Finally, we prove \cref{prop:equivalence-group-cvar}.
\begin{proof}
Write $r_b:=((\mathrm{Id}-B)^{-1})^T(-b^T,1,0)^T$ and $\Sigma_\varepsilon:=\Cov(\varepsilon\mid A)$.
The structural model gives
$Y-X^T b=r_b^T \varepsilon+r_b^T MA$, $P_A(Y-X^Tb)=r_b^T MA$, and $M_A(Y-X^Tb)=r_b^T\varepsilon$.
Then
$$
\BE[(Y-X^Tb)^2\mid A] = (r_b^TMA)^2 + r_b^T\Sigma_\varepsilon r_b = \BE[(M_A(Y-X^Tb))^2]+(P_A(Y-X^Tb))^2.
$$
Taking CVaR and using translation equivariance proves the identity.
\end{proof}

\subsection{Consistency and asymptotic normality}
\label{app:proof-asymptotic-normality}
\begin{proof}[Proof of \cref{prop:asymptotic-normality}]
Write the population projection coefficients as
$$
\pi_Y=\BE[AA^T]^{-1}\BE[AY],\qquad
\Pi_X=\BE[AA^T]^{-1}\BE[AX^T].
$$
Their sample versions converge by the law of large numbers. Together with consistency of empirical CVaR, this gives $\hat L(b)\toP L(b)$ for every $b$. Convexity makes this convergence uniform on compact sets, and uniqueness of the population minimizer gives consistency \citep[Section~1]{hjort1993asymptotics}.

For asymptotic normality, first consider (i) and (iii).
Under (i), $\hat L(b)$ is quadratic.
Under (iii), for $n$ large enough such that $\max_e |\hat p_e - p_e|$ is small enough ($\hat p_e = n_e/n$), around $b^\star$ the loss is essentially
$$
\| M_Z (Y - X b) \|_2^2 /n + \gamma ( \sum_{e < k} \frac{\hat p_e}{\alpha} \hat\mu_e(b)^2 + \frac{\tilde p_k}{\alpha} \hat\mu_k(b)^2 ),
$$
which is quadratic in $b$.
Here $0 < \tilde p_k = \alpha - \sum_{e < k} \hat p_e  < \hat p_k$ with probability tending to one.
These inequalities are strict by the assumption (iii) for $n$ large enough and $| \hat p_e - p_e|$ preserve the expression.
Thus, both for (i) and (iii) the loss is quadratic in some neighborhood of $b^\star$.
By consistency, $\hat b$ lies in this neighborhood of $b^\star$ for $n$ large enough and thus 
$\sqrt n(\hat b-b^\star) =- [\nabla_b^2\hat L(b^\star)]^{-1} \sqrt n \nabla_b\hat L(b^\star)$.
The Hessian is consistent with positive definite limit.
The gradient is a differentiable function of empirical first and second moments and so is asymptotically normal by the delta method.
Asymptotic normality of the estimator follows from Slutsky's theorem.

It remains to consider (ii), with $\gamma>0$ and $\alpha<1$.
By the Rockafellar--Uryasev representation, the CVaR anchor regression estimator minimizes the sample average
$$\hat L(b, t)
=\frac1n\sum_{i=1}^n\left[
\left(Y_i-X_i^Tb-A_i^T(\hat\pi_Y-\hat\Pi_Xb)\right)^2
+\gamma\left\{t+\frac1\alpha\left((A_i^T(\hat\pi_Y-\hat\Pi_Xb))^2-t\right)_+\right\}
\right].
$$
If the loss $\hat L$ used the population projections $\pi_Y, \Pi_X$ instead of the sample projections $\hat\pi_Y, \hat\Pi_X$, then the loss would be a sum of independent terms and asymptotic normality would follow from standard M-estimator results.
Instead, using the estimated projections $\hat \pi_Y, \hat \Pi_X$ introduces dependence of the individual loss terms.
The $\hat\pi_Y,\hat\Pi_X$ are least-squares coefficients and thus asymptotically normal.
Condition (ii) allows a first-order expansion in the projection coefficients.
Thus, to first order, estimating the projections adds a linear function of $\hat\pi_Y-\pi_Y$ and $\hat\Pi_X-\Pi_X$ to the empirical gradient.
To first order, the empirical loss' gradient is thus (a) the empirical loss' gradient using population projections (asymptotically normal as sum of i.i.d.\ terms), plus (b) a linear function of the estimation errors $\hat\pi_Y-\pi_Y$ and $\hat\Pi_X-\Pi_X$ (also asymptotcally normal).
The terms (a) and (b) are asymptotically jointly normal and so is their sum.
Under (ii), the population loss is twice continuously differentiable with positive definite Hessian.
Asymptotic normality of $(\hat b, \hat t)$ follows from Slutsky's theorem.
\end{proof}

\subsection{Sampling variance with a scalar Gaussian anchor}
\label{app:proof-fixed-anchor-sampling-variance}
\begin{proof}[Proof for \cref{ex:4}]
We have $X=A+H+\varepsilon_X$ and $Y=\beta X+H+\varepsilon_Y$, $H, A, \varepsilon_X, \varepsilon_Y \sim \CN(0,1)$ i.i.d.
Asymptotic normality follows from \cref{prop:asymptotic-normality}.
We compute the limiting variance.

The population residual mean and variance given the anchor are
$$
\BE[Y-Xb\mid A]=(\beta-b)A, \quad \Var(Y-Xb\mid A)=(1+\beta-b)^2+(\beta-b)^2+1.
$$
Thus, using \cref{prop:elliptical},
$$
L(b)=(1+\beta-b)^2+(1+\gamma c_\alpha)(\beta-b)^2+1
\quad\Rightarrow\quad
b_\mathrm{cvar-anchor}^{\gamma,\alpha}=\beta+\frac1{2+\gamma c_\alpha}.
$$

Write
$$
r=Y-Xb, \quad
\hat\pi_X=\frac{\sum_i A_iX_i}{\sum_i A_i^2},\quad
\hat\pi_r=\frac{\sum_i A_ir_i}{\sum_i A_i^2},\quad
\hat c_\alpha=\widehat{\cvar}_\alpha((A_i^2)_{i=1}^n).
$$
We have $(\hat P_A r)_i = A_i \hat \pi_r$ and so $\gamma\widehat{\cvar}_\alpha((\hat P_A r)^2) = \gamma \hat \pi_r^2 \hat c_\alpha$.
The empirical loss is
$$
\hat L(b) = \| \hat M_A (Y - X b) \|_2^2/n + \gamma \widehat{\cvar}_\alpha((\hat P_A(Y-Xb))^2) = \frac{1}{n} \sum_i(r_i-\hat\pi_r A_i)^2 +\gamma\hat\pi_r^2\hat c_\alpha.
$$
We compute the derivative $\hat L'$ using $\frac{\dd r_i}{\dd b}=-X_i$ and $\frac{\dd\hat\pi_r}{\dd b}=-\hat\pi_X$.
$$
-\frac{1}{2}\hat L'(b) = \frac{1}{n} \sum_i (X_i - \hat \pi_X A_i)(r_i - \hat \pi_r A_i) + \gamma \hat \pi_X \hat \pi_r \hat c_\alpha,
\quad  \hat L''(b)
  =\frac2n\sum_i(X_i-\hat\pi_X A_i)^2
  +2\gamma\hat\pi_X^2\hat c_\alpha.
$$
Using $\hat \pi_X \toP 1$, $\frac{1}{n} \sum_i (X_i - \hat \pi_X A_i)^2 \toP \BE[(X - A)^2] = 2$, and $\hat c_\alpha \toP c_\alpha$ gives $\hat L''(b) \toP 2(2 + \gamma c_\alpha)$.

Next, we compute $\Var(\hat L'(b_\mathrm{cvar-anchor}^{\gamma,\alpha}))$.
Conditional on the anchors, the residual sample covariance is independent of the fitted projection coefficients.
Its conditional distribution does not depend on the anchors, so the two terms in the derivative are independent.
Using $\BE[X^2 Y^2] = \Var(X)\Var(Y)+2\Cov(X,Y)^2$ for centered Gaussian $X$, $Y$, yields
\begin{align*}
n\Var( \frac1n\sum_i(X_i-\hat\pi_X A_i)(r_i-\hat\pi_r A_i) )
&=\frac{n-1}{n} \left(\Var(M_AX)\Var(M_Ar)+\Cov(M_AX,M_Ar)^2\right) \\
&\longrightarrow \frac{5(\gamma c_\alpha)^2+12\gamma c_\alpha+12}{(2+\gamma c_\alpha)^2}.
\end{align*}
For the second term, we first condition on the values $(A_i)_i$ and thus also $\hat c_\alpha$.
\begin{align*}
\BE[\hat\pi_X\hat\pi_r\mid(A_i)_i] &=-\frac1{2+\gamma c_\alpha} +\frac{\gamma c_\alpha}{(2+\gamma c_\alpha)\sum_i A_i^2}
\end{align*}
Using $n^{-1}\sum_i A_i^2\toP1$ and $\hat c_\alpha\toP c_\alpha$, the variance of the conditional mean converges to
$$
n\Var\left( \BE[\hat\pi_X\hat\pi_r\hat c_\alpha\mid(A_i)_i] \right) \longrightarrow \frac{v_\alpha}{(2+\gamma c_\alpha)^2}.
$$
We continue with the conditional variance, again using $\BE[X^2 Y^2] = \Var(X)\Var(Y)+2\Cov(X,Y)^2$.
$$n\Var(\hat\pi_X\hat\pi_r\mid(A_i)_i)
=\frac{2n\{(\gamma c_\alpha)^2+2\gamma c_\alpha+4\}}
{(2+\gamma c_\alpha)^2\sum_i A_i^2}
+\frac{n\{5(\gamma c_\alpha)^2+12\gamma c_\alpha+12\}}
{(2+\gamma c_\alpha)^2(\sum_i A_i^2)^2}
\toP
\frac{2\{(\gamma c_\alpha)^2+2\gamma c_\alpha+4\}}
{(2+\gamma c_\alpha)^2}.
$$
Thus, using $n \Var(\hat c_\alpha)\to v_\alpha$ and the law of total variance, we have
\begin{align*}
n\Var(\hat\pi_X\hat\pi_r\hat c_\alpha)
&=n\BE[\Var(\hat\pi_X\hat\pi_r\hat c_\alpha\mid(A_i)_i)]
+n\Var(\BE[\hat\pi_X\hat\pi_r\hat c_\alpha\mid(A_i)_i])\\
&\longrightarrow
\frac{2c_\alpha^2 ((\gamma c_\alpha)^2+2\gamma c_\alpha+4 )}{(2+\gamma c_\alpha)^2} + \frac{v_\alpha}{(2+\gamma c_\alpha)^2}.
\end{align*}

Thus
$$
n \Var(\hat L'(b)) \longrightarrow \frac{4}{(2+\gamma c_\alpha)^2}\left[5(\gamma c_\alpha)^2+12\gamma c_\alpha+12 + 2(\gamma c_\alpha)^2 ((\gamma c_\alpha)^2+2\gamma c_\alpha+4) + \gamma^2v_\alpha\right].
$$
Dividing by $\plim (\hat L''(b))^2 = [2(2+\gamma c_\alpha)]^2$ gives
$$
V^{\gamma,\alpha}
=
\frac{
2(\gamma c_\alpha)^4+4(\gamma c_\alpha)^3
+13(\gamma c_\alpha)^2+12\gamma c_\alpha+12
+\gamma^2v_\alpha
}{(2+\gamma c_\alpha)^4}.
$$
\end{proof}

\subsection{Proofs of appendix \ref{sec:intercepts}}
Below is the proof of \cref{prop:intercepts}.
\begin{proof}
Let $r_b=((\Id-B)^{-1})^T(-b^T,1,0)^T$.
As $P_A1=1$ and $\BE[A\varepsilon^T]=0$, it follows that
$P_A r=r_b^TMA-b_0$ and $M_A r=r_b^T\varepsilon$.
Fix some $\BQ \in \CQ_\alpha$.
Under $P_w$ for $w=\BE_\BQ[MA]+v$, the residual is
$$r=r_b^T \varepsilon+r_b^T w-b_0 = r_b^T\varepsilon+ r_b^T\BE_\BQ[MA]-b_0 + r_b^Tv = M_A r + \BE_\BQ[P_A r] + r_b^T v.$$
We take expectations of the square and expand the square.
As $\varepsilon$ has mean zero and is independent of $v$, the cross terms involving $\BE_w[M_A r] = \BE[M_A r] = 0$ vanish.
$$
\BE_w[r^2] = \BE[(M_A r)^2] + \left( \BE_\BQ[P_A r] \right)^2 + \BE_w[(r_b^Tv)^2] + 2\BE_\BQ[P_A r]\,\BE_w[r_b^Tv].
$$
Let $K:=\gamma M\Cov_\BQ(A)M^T$ and $\BE_w[vv^T]\preceq K$.
Then
$$ \BE_w[(r_b^T v)^2]\leq r_b^TKr_b =\gamma \, r_b^T M\Cov_\BQ(A)M^T r_b =\gamma\Var_\BQ(P_A r).$$

For the first equation with $v \in C_1$, we apply Cauchy--Schwarz:
$|\BE_w[r_b^Tv]|\leq\sqrt{\BE_w[(r_b^Tv)^2]}\leq\sqrt{r_b^TKr_b}$.
Substituting these bounds into the expansion gives
$$
\BE_w[r^2]\leq\BE[(M_A r)^2]+\left(|\BE_\BQ[P_A r]|+\sqrt{\gamma\Var_\BQ(P_A r)}\right)^2.
$$
This is attained at the deterministic choice $v=\pm Kr_b/\sqrt{r_b^TKr_b}$ if $r_b^TKr_b>0$, or $v=0$ otherwise.
Choose the sign so $r_b^T v$ has the same sign as $\BE_\BQ[P_A r]$.
This yields 
$$
\BE_w[r^2] = \BE[(M_A r)^2] +\left( | \BE_\BQ[P_A r] | + \sqrt{ \gamma \Var_\BQ(P_A r) } \right)^2.$$
Taking the supremum over $\BQ \in \CQ_\alpha$ yields the first equality.

For the second equation with $v \in C_2$, we have $\BE_w[r_b^Tv]=0$ and
$$
\BE_w[r^2] \leq \BE[(M_A r)^2] + \bigl(\BE_\BQ[P_A r]\bigr)^2 + \gamma\Var_\BQ(P_A r)
$$
follows. To show that this is attained, let $v$ be a Gaussian with mean zero and covariance $K$, independent of $\varepsilon$. Then $\BE_w[r_b^Tv]=0$ and $\BE_w[(r_b^Tv)^2]=r_b^TKr_b=\gamma\Var_\BQ(P_A r)$.
Taking the supremum over $\BQ \in \CQ_\alpha$ yields the second equality.

We now show the equivalences of minimizers and minimum values.
For (i), $\CQ = \{ \BP_\mathrm{train} \}$ and $\gamma > 0$.
For fixed $b$ the terms $\BE[(M_A r)^2]$ and $\Var(P_A r)$ (as $P_A 1 = 1$) do not depend on $b_0$.
The minimizing $b_0^\star(b) = \BE[P_A (Y - X^T b)] = \BE[Y - X^T b]$ (as $P_A 1 = 1$) such that $\BE[P_A (Y - X^T b - b_0^\star(b))] = \BE[Y - X^T b - b_0^\star(b)] = 0$ and the minimum value is $\BE[(M_A (Y - X^T b - b_0^\star(b)))^2]+\gamma\Var(P_A (Y - X^T b - b_0^\star(b)))$ for all three losses.
Thus the minimizers and minimizing values are the same.

For (ii), again, we fix $b$ and minimize over $b_0$.
The two losses are
\begin{align*}
&\mathrm{loss}_1(b, b_0) = \BE[(M_A r)^2] + \sup_{\BQ \in \CQ} \{ \gamma \Var_\BQ(P_A r) + \gamma \BE_\BQ[P_A r]^2 \} \quad \text{and} \\
& \mathrm{loss}_2(b, b_0) = \BE[(M_A r)^2] + \sup_{\BQ \in \CQ} \{ \gamma \Var_\BQ(P_A r) + \BE_\BQ[P_A r]^2 \} 
\end{align*}
As $P_A 1 = 1$, both terms $\BE[(M_A r)^2]$ and $\Var_\BQ(P_A r)$ do not depend on $b_0$.
Only $\BE_\BQ[P_A r]^2$ does.

As $\CQ_\alpha$ is compact and $\BQ \mapsto \Var_\BQ(P_A (Y - X^T b))$ is continuous, the supremum $\sup_{\BQ \in \CQ_\alpha} \Var_\BQ(P_A (Y - X^T b))$ is attained at some $\BQ^\star(b) \in \CQ_\alpha$.
Define $b_0^\star(b) := \BE_{\BQ^\star(b)}[P_A (Y - X^T b)]$ and $r(b) := Y - X^T b - b_0^\star(b)$. Then $\BE_{\BQ^\star(b)}[P_A r(b)] = 0$.
Let $\BQ \in \CQ_\alpha$. We prove that
\begin{align*}
& \text{(i)} \quad \gamma \Var_\BQ(P_A r(b)) + \gamma \BE_\BQ[P_A r(b)]^2 \leq \gamma \Var_{\BQ^\star(b)}(P_A r(b)) + \gamma \BE_{\BQ^\star(b)}[P_A r(b)]^2 = \gamma \Var_{\BQ^\star(b)}(P_A r(b)), \\
& \text{(ii)} \quad \gamma \Var_\BQ(P_A r(b)) + \BE_\BQ[P_A r(b)]^2 \leq \gamma \Var_{\BQ^\star(b)}(P_A r(b)) + \BE_{\BQ^\star(b)}[P_A r(b)]^2 = \gamma \Var_{\BQ^\star(b)}(P_A r(b)).
\end{align*}
For any $t \in [0, 1]$, define $\BQ_t := t \BQ + (1 - t) \BQ^\star(b) \in \CQ_\alpha$.
Then
\begin{multline*}
\Var_{\BQ_t}(P_A r(b)) =
\BE_{\BQ_t}[(P_A r(b))^2] - (\BE_{\BQ_t}[P_A r(b)])^2 \\
= t \BE_\BQ[(P_A r(b))^2] + (1-t) \BE_{\BQ^\star}[(P_A r(b))^2] - (t \BE_\BQ[P_A r(b)] + (1-t) \BE_{\BQ^\star}[P_A r(b)])^2 \\
 = t \BE_\BQ[(P_A r(b))^2] + (1-t) \BE_{\BQ^\star}[(P_A r(b))^2] - t^2 \BE_\BQ[P_A r(b)]^2 \\
\Rightarrow 0 \geq \left. \frac{\dd}{\dd t} \Var_{\BQ_t}(P_A r(b)) \right|_{t=0} = \BE_\BQ[(P_A r(b))^2] - \BE_{\BQ^\star}[(P_A r(b))^2] = \BE_\BQ[(P_A r(b))^2] - \Var_{\BQ^\star}(P_A r(b)).
\end{multline*}
The inequality follows from $\BQ^\star(b)$ maximising $\Var_\BQ(P_A r(b))$ over $\CQ_\alpha$.
Thus, for $\gamma \geq 1$,
$$
\gamma \Var_\BQ(P_A r(b)) + \BE_\BQ[P_A r(b)]^2 \leq \gamma \Var_\BQ(P_A r(b)) + \gamma \BE_\BQ[P_A r(b)]^2 = \gamma \BE_\BQ[(P_A r(b))^2] \leq \gamma \Var_{\BQ^\star(b)}(P_A r(b)).
$$
This shows (i) and (ii) above.
This implies that
$$
\min_{b_0} \mathrm{loss}_1(b, b_0) = \BE[(M_A (Y - X^T b))^2]+ \gamma\Var_{\BQ^\star}(P_A (Y - X^T b)) = \min_{b_0} \mathrm{loss}_2(b, b_0)
$$ with minimizer $b_0^\star(b)$. Thus also the minimizer over $b$ and minimum value are the same for both losses.
\end{proof}

\end{document}